\documentclass[sigplan,screen,authorversion]{acmart}
\pdfoutput=1
\AtBeginDocument{%
  \providecommand\BibTeX{{%
    \normalfont B\kern-0.5em{\scshape i\kern-0.25em b}\kern-0.8em\TeX}}}

\copyrightyear{2022} 
\acmYear{2022} 
\setcopyright{rightsretained} 
\acmConference[LICS '22]{37th Annual ACM/IEEE Symposium on Logic in Computer Science (LICS)}{August 2--5, 2022}{Haifa, Israel}
\acmBooktitle{37th Annual ACM/IEEE Symposium on Logic in Computer Science (LICS) (LICS '22), August 2--5, 2022, Haifa, Israel}
\acmDOI{10.1145/3531130.3533369}
\acmISBN{978-1-4503-9351-5/22/08}

\usepackage{ltl}
\usepackage{tikz}
\usepackage{stmaryrd}
\usepackage{cleveref}
\usepackage{todonotes}
\usepackage{algorithm}
\usepackage{multicol}
\usepackage{multirow}
\usepackage[noend]{algpseudocode}
\usepackage{thm-restate} 
\usepackage{stmaryrd}
\usepackage{xspace}
\usepackage{xcolor}
\usepackage{subcaption}
\usepackage{ marvosym }
\usepackage{enumitem}
\usepackage{pifont}% http://ctan.org/pkg/pifont
\newcommand{\cmark}{\ding{51}}%
\newcommand{\xmark}{\ding{55}}%
\usepackage{booktabs}
\usepackage{bussproofs}
\usepackage{thm-restate} 

\definecolor{pastelblue}{HTML}{2874ae}
\definecolor{chestnut}{HTML}{A24516}

\newcommand{\ltlN}{\LTLnext}
\newcommand{\ltlG}{\LTLglobally}
\newcommand{\ltlF}{\LTLeventually}
\newcommand{\ltlU}{\LTLuntil}

\newcommand{\ap}{\mathit{AP}}

\newcommand{\nat}{\mathbb{N}}
\newcommand{\set}[1]{\{ #1 \}}
\newcommand{\pow}[1]{2^{#1}}
\newcommand{\ldot}{\mathpunct{.}}
\newcommand{\lang}[1]{\mathcal{L}(#1)}
\newcommand{\ltlmodels}{\models}

\newcommand{\aut}{\mathcal{A}}
\newcommand{\pushAut}{\mathcal{P}}
\newcommand{\system}{\mathcal{S}}

\newcommand{\traceVars}{\mathcal{V}}
\newcommand{\traceSet}{T}
\newcommand{\finTraceSet}{U}

\newcommand{\PSPACE}{\texttt{PSPACE}\@\xspace}
\newcommand{\EXPSPACE}{\texttt{EXPSPACE}\@\xspace}
\newcommand{\NEXPT}{\texttt{NEXPTIME}\@\xspace}
\newcommand{\NP}{\texttt{NP}\@\xspace}
\newcommand{\coRE}{\texttt{coRE}\@\xspace}

\newcommand{\nstep}[2]{\mathit{Step}_{#1}(#2)}
\newcommand{\nnstep}[2]{\mathit{Reach}_{#1}(#2)}

\newcommand{\forward}{\Rrightarrow}
\newcommand{\backward}{\Lleftarrow}

\def\proofsnamefont{\itshape}
\def\proofsindent{\noindent}

\newenvironment{proofSketch}[1][Proof sketch]{\par
    \pushQED{\qed}%
    \normalfont %\topsep6\p@\@plus6\p@\relax
    \trivlist
    \item[\proofsindent\hskip\labelsep
    {\proofsnamefont #1.}]\ignorespaces
}{%
    \popQED\endtrivlist%\@endpefalse
}

\newcommand\xqed[1]{%
    \leavevmode\unskip\penalty9999 \hbox{}\nobreak\hfill
    \quad\hbox{#1}}
\newcommand\demo{\xqed{$\vartriangleleft$}}

\AtEndPreamble{%
	\theoremstyle{acmdefinition}
	\newtheorem{remark}[theorem]{Remark}}

\newif\iffullversion
\fullversiontrue

\newcommand{\ifFull}[1]{\iffullversion#1\else the full version \cite{fullVersion}\fi}

\begin{document}

\title[Deciding Hyperproperties Combined with Functional Specifications]{Deciding Hyperproperties Combined \\with Functional Specifications}

\author{Raven Beutner}
\orcid{0000-0001-6234-5651}
\affiliation{%
	\institution{CISPA Helmholtz Center for Information Security}
	\country{Germany}
}

\author{David Carral}
\orcid{0000-0001-7287-4709}
\affiliation{%
	\institution{LIRMM, Inria, University of Montpellier, CNRS}
	\country{France}
}

\author{Bernd Finkbeiner}
\orcid{0000-0002-4280-8441}
\affiliation{%
	\institution{CISPA Helmholtz Center for Information Security}
	\country{Germany}
}

\author{Jana Hofmann}
\orcid{0000-0003-1660-2949}
\affiliation{%
	\institution{CISPA Helmholtz Center for Information Security}
	\country{Germany}
}

\author{Markus Krötzsch}
\orcid{0000-0002-9172-2601}
\affiliation{%
	\institution{Technische Universität Dresden}
	\country{Germany}
}

\renewcommand{\shortauthors}{Beutner et al.}

\begin{abstract}
	We study satisfiability for HyperLTL with a $\forall^*\exists^*$
	quantifier prefix, known to be highly undecidable in general.
	HyperLTL can express system properties that relate multiple traces (so-called \emph{hyperproperties}), which are often
	combined with \emph{trace properties} that specify functional behavior on single traces.
	Following this conceptual split, we first define several \emph{safety} and \emph{liveness} fragments
	of $\forall^*\exists^*$ HyperLTL, and characterize the complexity of their (often much easier) satisfiability problem.
	We then add LTL trace properties as functional specifications.
	Though (highly) undecidable in many cases, this way of combining ``simple'' HyperLTL and arbitrary
	LTL also leads to interesting new decidable fragments.		
	This systematic study of $\forall^*\exists^*$ fragments is complemented by a new
	(incomplete) algorithm for $\forall\exists^*$-HyperLTL satisfiability.
\end{abstract}




\maketitle

\section{Introduction}

Hyperproperties are properties that relate multiple execution traces of a system \cite{ClarksonS08} and comprise a range of relevant properties from many areas of computer science. Examples are symmetry, optimality, robustness, and noninterference.
The most prominent logic for expressing hyperproperties is HyperLTL \cite{ClarksonFKMRS14}, which extends LTL with trace quantification.
Generalized noninterference \cite{McCullough88}, for example, states that high-security inputs do not influence the input-output behavior observable by a low-security user, which can be expressed in HyperLTL as follows.
\[
	\forall \pi  \forall \pi'  \exists \pi'' \ldot \LTLglobally \Big(\bigwedge_{a \in L_\mathit{out} \cup L_\mathit{in}} \!\!\!\!\!(a_\pi \leftrightarrow a_{\pi''}) \land \!\!\bigwedge_{a \in H_\mathit{in}} (a_{\pi'} \leftrightarrow a_{\pi''} ) \Big)
\]
The formula states that for every two traces $\pi, \pi'$, there exists a trace $\pi''$ that combines the low-security inputs and outputs on $\pi$ and the high-security inputs on $\pi'$.

In this paper, we study the satisfiability problem of HyperLTL.
For LTL, satisfiability is \PSPACE-complete~\cite{sistla1985complexity}.
For hyperproperties, satisfaction cannot be decided by analyzing single traces in isolation, making formal reasoning challenging. 
Deciding satisfiability in the $\exists^*\forall^*$ fragment of HyperLTL is already \EXPSPACE-complete~\cite{decidable-hyperltl-1}; and deciding hyperproperties with a $\forall^*\exists^*$ trace quantifier alternations is, in general, strongly undecidable, namely $\Sigma_1^1$-complete~\cite{FortinKT021}.
The $\forall^*\exists^*$ fragment contains many relevant properties like generalized noninterference, program refinement, and software doping~\cite{McCullough88,DArgenioBBFH17}.
However, positive results for this important fragment have been very rare and were only obtained by heavy restrictions on the use of temporal operators or by assuming finite models~\cite{decidable-hyperltl-2} (see related work below). 
Algorithms, even if incomplete, are similarly missing.

In this work, we address these shortcomings by studying ways of solving
satisfiability of $\forall^*\exists^*$ HyperLTL specifications.
We identify simple yet expressive fragments of $\forall^*\exists^*$ 
with better computational properties, where
our approach derives interesting fragments in two steps.
First, we split a specification into hyperproperty and trace property, so that we can focus on ``simple'' hyperproperties.
Second, to find such simple hyperproperties, we systematically study fragments of
temporal safety and temporal liveness hyperproperties.
This work towards new decidable fragments is complemented by 
a new (incomplete but often successful) algorithm that is applicable to arbitrary $\forall\exists^*$ specifications.

\paragraph{Splitting in Hyperproperties and Trace Properties.} 
So far, all HyperLTL decidability results were obtained by considering HyperLTL specifications in isolation. Most of the time, however, specifications refer to a specific system. The hyperproperty itself is often relatively simple (like the noninterference property above) and only gets difficult to satisfy given a specification of the functional behavior of the system.\footnote{Of course, we can incorporate the LTL property in the HyperLTL formula: we conceptually divide the specification into a (complicated) trace property and a (simple) hyperproperty.}
The following example highlights this interplay between functional property and hyperproperty.
\begin{example}\label{ex:introEx}
	Consider a system of agents that send and receive data.
	Each trace describes the behavior of a single agent.
	We want the system to satisfy the following hyperproperty.
	\begin{align*}
		\varphi \coloneqq \forall \pi \exists \pi' \ldot \ltlF (\mathit{send}_\pi \land \mathit{rec}_{\pi'})
	\end{align*}%
	The formula states that each agent eventually sends its information and that there exists some agent receiving it. The formula on its own is easily satisfiable already by a one-trace model.
    In addition to the hyperproperty we add the simple functional specification (trace property)
	\begin{align*}
		\psi \coloneqq (\neg \mathit{rec}) \ltlU (\mathit{rec} \land \ltlN \ltlG \neg \mathit{rec}) \land \ltlG (\mathit{rec} \leftrightarrow \ltlN \mathit{send})
	\end{align*}
	which expresses that each agent receives data exactly once and sends it forth in the next step.
	Every model that satisfies the \emph{combination} of $\varphi$ and $\psi$ needs to be infinite. 
	Automatically checking satisfiability is thus complex as we cannot iteratively search for models of bounded (finite) size. 
	\demo
\end{example}
A satisfiability checker that distinguishes between a functional specification and hyperproperties could be used to sanity-check whether a hyperproperty is satisfiable in combination with the specification of the system at hand.

\paragraph{Temporal Safety and Temporal Liveness.}
The classification into safety and liveness has a long tradition in the study of trace properties,
where especially safety often allows for easier algorithms.
For our analysis, we define analogous fragments:
a HyperLTL formula is \emph{temporal safety} (resp. \emph{temporal liveness}) if its LTL body describes a safety (resp. liveness) property.
We study the relationship to the existing notations of \emph{hypersafety} and \emph{hyperliveness} defined by Clarkson and Schneider \cite{ClarksonS08}.
Guided by our insights into the complete fragments, we derive several more specific classes of temporal safety and liveness
properties, for which satisfiability is easier to decide.

\begin{table}
	\small
	\caption{Deciding satisfiability of HyperLTL specifications.
	    All results, expect for decidability (dec.), denote completeness.
		Our notation is found in \Cref{subsec:notation}, e.g., $ \forall^*\exists^* \ldot {\protect\LTLglobally} ({\protect\LTLnext}^*)$ is the class of $\forall^*\exists^*$ formulas whose LTL body uses a single ${\protect \LTLglobally}$ operator with optional ${\protect \LTLnext}$ operators in its scope.}
	\begin{tabular}{@{~}c@{\hspace{2mm}}c@{\hspace{2mm}}r@{~}l@{\hspace{2mm}}r@{~}l@{~}}
		\toprule
		&& \multicolumn{2}{@{}c@{}}{\textbf{no LTL spec.}} & \multicolumn{2}{@{}c@{}}{\textbf{with LTL spec.}} \\
		\midrule
		& complete fragment & \coRE&[Thm.~\ref{theo:coReComp}]  & $\Sigma_1^1$&[Thm.~\ref{theo:safetyUndec}]  \\[1pt]
		& $\forall^*\exists^* \ldot \LTLnext^*$ & \texttt{NEXP}&[Thm.~\ref{theo:prop}] & \texttt{NEXP}& [Thm.~\ref{theo:prop}] \\[1pt]
		& $\forall^*\exists^* \ldot \LTLglobally$ & \texttt{NEXP}&[Lem.~\ref{lem:saftyDecNexptime}]& $\Sigma_1^1$&[Thm. \ref{theo:safetyUndec}] \\[1pt]
	\rotatebox{90}{\makebox[0pt][l]{\textbf{\begin{minipage}{10ex}\centering Temporal\\Safety\end{minipage}}}}
		& $\forall^*\exists^* \ldot \LTLglobally(\LTLnext^*)$ & \coRE&[Lem.~\ref{lem:coReHardness}] & $\Sigma_1^1$&[Thm.~\ref{theo:safetyUndec}] \\
		\midrule
		&complete fragment & $\Sigma_1^1$&[Thm.~\ref{theo:livenessAnaytical}] & $\Sigma_1^1$&[Thm.~\ref{theo:livenessAnaytical}] \\[1pt]
		&$\forall \exists^*\ldot \text{det-liveness}$ & trivial&[Prop.~\ref{prop:alwaysSat}] & $\Sigma_1^1$&[Cor.~\ref{cor:detLiveness}] \\[1pt]
		&$\forall \exists^*\ldot \LTLeventually (\LTLnext^*) $ & \NP&[Lem.~\ref{lem:npEventual}]& dec.&[Thm.~\ref{theo:eventuallyDec}] \\[1pt]
	\rotatebox{90}{\makebox[0pt][l]{\textbf{\begin{minipage}{10ex}\centering Temporal\\Liveness\end{minipage}}}}
		& $\forall^* \exists^*\ldot \LTLeventually \land \cdots \land \LTLeventually $ & \NP& [Lem.~\ref{lem:npEventual}] & $\Sigma_1^1$& [Thm.~\ref{theo:conjunctionFAnalytical}] \\
		\bottomrule
	\end{tabular}	
	\label{tab:results}
\end{table}
\paragraph{Main Results.} Our results are summarized in Table~\ref{tab:results},
where each line represents a class of HyperLTL properties, and the columns distinguish
whether or not additional (arbitrarily complex) LTL specifications are allowed.
All hardness results for $\forall^*\exists^*$ fragments, except in the \texttt{NEXP} cases,
already hold for $\forall\exists^*$.
The restriction to temporal safety makes the satisfiability of HyperLTL drop from $\Sigma_1^1$ to \coRE,
which we show by an effective reduction to satisfiability of first-order logic. While still undecidable, this enables the
use of common first-order techniques such as resolution, tableaux, and related methods~\cite{arhandbook}. %\todo{R:Are these cites good?}
Hardness already holds for simple formulas consisting only of a single $\LTLglobally$ with $\LTLnext$s in its scope.
If we add (non-safety) functional specifications, hardness jumps back to $\Sigma^1_1$.
In contrast to temporal safety properties, the class of temporal liveness HyperLTL formulas is of analytical complexity, even without additional LTL specifications. However, again in contrast to $\LTLglobally(\LTLnext^*)$, formulas from the $\forall\exists^* \ldot \LTLeventually(\LTLnext^*)$ fragment are decidable, even when combined with an arbitrary LTL specification. 
This is the first HyperLTL decidability result for formulas that can enforce models with infinitely many traces.
The class also contains the specification from Example~\ref{ex:introEx}. 
This decidability result is tight in the sense that already conjunctions of multiple eventualities are analytical again.

Finally, to complement our decidability results, we propose a general approximation algorithm to find the \emph{largest} model for specifications consisting of a HyperLTL formula and an LTL formula. 
Our experimental evaluation shows that our algorithm performs significantly better than approaches that iteratively search for models of bounded size \cite{decidable-hyperltl-2,FinkbeinerHH18} and can even show unsatisfiability for many formulas (which is impossible in bounded approaches). 

\paragraph{Structure}

The remainder of this paper is structured as follows.
We give some basic preliminaries and introduce HyperLTL in \Cref{sec:prelim}. 
We study the fragment of temporal safety in \Cref{sec:safety}. 
We begin this study with the full fragment, and then gradually decrease in expressiveness all the way to the fragment containing only $\ltlN$ operators. 
We then move to temporal liveness in \Cref{sec:livness}.
Analogous to the safety case, we begin with the full fragment, gradually decreasing to the fragment of pure eventualities, for which we establish decidability.
Finally, in \Cref{sec:largestModelsAlg}, we describe our approximation for finding the largest models and report on experimental results in \Cref{sec:eval}.

\paragraph{Related Work.}
In recent years, many logics to express hyperproperties have been developed.
Most approaches extend existing logics with trace or path quantification, examples besides HyperLTL are HyperCTL$^*$ \cite{ClarksonS08}, HyperQPTL \cite{Rabe16}, HyperPDL-$\Delta$~\cite{GutsfeldMO20}, and HyperATL$^*$ \cite{BeutnerF21}.
Monadic first-order logics can be extended by adding a special equal-level predicate \cite{Finkbeiner017} or using different types of quantifiers~\cite{flavoursSequential}.
Recently, hyperproperties have also been obtained via a team semantics for trace logics \cite{KrebsMV018,VirtemaHFK021}.
Apart from plain temporal logics, there are also hyperlogics for hyperproperties that are asynchronous~\cite{DBLP:conf/cav/BaumeisterCBFS21, DBLP:conf/lics/BozzelliPS21, DBLP:journals/pacmpl/GutsfeldMO21}, quantitative~\cite{FinkbeinerHT18}, or probabilistic~\cite{AbrahamBBD20,DimitrovaFT20}.

HyperLTL remains the most used among the proposed hyperlogics.
Its satisfiability problem is known to be challenging:
if we define fragments based on quantifier prefixes (but with an arbitrary body),
then $\exists^*\forall^*$ is the most general fragment for which satisfiability
is still decidable (and \texttt{EXPSPACE}-complete), whereas
$\forall\exists^*$ already leads to undecidability \cite{decidable-hyperltl-1}.
In fact, the $\forall^*\exists^*$ fragment is already satisfiability-complete: any HyperLTL formula can be effectively translated into equisatisfiable $\forall^*\exists^*$ formula \cite{decidable-hyperltl-2}.
Analyzing the case of (unrestricted) HyperLTL in more detail,
Fortin et al. show satisfiability to be $\Sigma_1^1$-complete,
and therefore above all problems in the arithmetic hierarchy~\cite{FortinKT021}. 
In a more fine-grained analysis, Mascle and Zimmermann show that the problem becomes decidable
if one only considers models of a bounded size or if, for selected quantifier prefixes, temporal operators
are not nested~\cite{decidable-hyperltl-2}.
In particular, $\forall\exists^*$ properties using only $\ltlF$ and $\ltlG$ (without $\ltlN$s) are decidable (and always have a finite model), as no ``diagonal'' comparison between trace positions is possible~\cite{decidable-hyperltl-2}. 
The satisfiability of the logics HyperQPTL and HyperCTL$^*$, which both subsume HyperLTL, has been studied as well~\cite{CoenenFHH19}.

\section{Preliminaries}\label{sec:prelim}
We assume a fixed, finite set of atomic propositions $\ap$ and write $\Sigma \coloneqq  2^\ap$.
Given a symbol $\pi$, we write $\ap_\pi$ for the set $\set{a_\pi \mid a \in \ap}$.
A trace $t$ is an element in $\Sigma^\omega$.
For $i \in \nat$, $t(i)$ denotes the $i$th element in $t$ (starting with the $0$th) and $t[i, \infty]$ is the suffix of a trace starting in point in time $i$.
For a finite trace $u \in \Sigma^*$ and an infinite trace $t \in \Sigma^\omega$, $u$ is a prefix of $t$ (written $u \lessdot t$) if for every $0 \leq i < |u|$, $u(i) = t(i)$.
A trace property $P$ is a set of traces, whereas a hyperproperty $H$ is a set of sets of traces~\cite{ClarksonS08}.

\subsection{Trace Properties and LTL}
Linear temporal logic (LTL) defines trace properties by combining temporal operators with boolean connectives. Its syntax is defined by the following grammar.
\begin{align*}
    \psi &\coloneqq a \mid \neg \psi \mid \psi \land \psi \mid \ltlN \psi \mid \psi \ltlU \psi
\end{align*}%
where $a \in \ap$. We also use the standard Boolean connectives $\wedge$, $\rightarrow$, $\leftrightarrow$ and constants $\top, \bot$, as well as the derived LTL operators \emph{eventually} $\ltlF \psi \coloneqq \top \ltlU \psi $, and \emph{globally} $\ltlG \psi \coloneqq \neg \ltlF \neg \psi$.
The semantics of LTL is defined as usual.
\begin{align*}
    t &\ltlmodels  a &\text{iff} \quad  &a \in t(0)\\
     t &\ltlmodels  \neg \psi &\text{iff} \quad & t \not\ltlmodels  \psi \\
    t &\ltlmodels  \psi_1 \land \psi_2 &\text{iff} \quad  &t \models \psi_1 \text{ and }  t \models  \psi_2\\
    t &\ltlmodels  \ltlN  \psi &\text{iff} \quad & t[1,\infty] \ltlmodels \psi \\
    t &\ltlmodels  \psi_1 \ltlU \psi_2 &\text{iff} \quad & \exists i \ldot t[i, \infty]\ltlmodels  \psi_2 \text{ and } \forall j < i \ldot  t[j,\infty] \ltlmodels  \psi_1
\end{align*}%

Safety and liveness properties are prominent classes of trace properties~\cite{AlpernS85}.  
Safety properties are characterized by the fact that each violation is caused after a finite time. Liveness properties characterize that something good happens eventually. 
\begin{definition}\label{def:safetyLiveness}
    A property $P$ is \emph{safety} if it holds that for every trace $t \notin P$, there exists a $u \lessdot t$ such that for every $t'$ with $u \lessdot t'$, we have $t' \notin P$.
    A property $P$ is \emph{liveness} if for every $u \in \Sigma^*$, there exists a $t \in \Sigma^\omega$ with $u \lessdot t$ and $t \in P$.
\end{definition}

\subsection{Hyperproperties and HyperLTL}

HyperLTL \cite{ClarksonFKMRS14} extends  LTL with explicit quantification over traces, thereby lifting it from a logic expressing trace properties to one expressing hyperproperties \cite{ClarksonS08}. 
Let $\traceVars$ be a set of trace variables. 
We define HyperLTL formulas with the following grammar. 
\begin{align*}
    \varphi &\coloneqq \exists \pi \ldot \varphi \mid \forall \pi \ldot \varphi \mid \phi \\
    \phi &\coloneqq a_\pi \mid \neg \phi \mid \phi \land \phi \mid \ltlN \phi \mid \phi \ltlU \phi
\end{align*}%
Here, $\pi \in \traceVars$ and $a \in \ap$. 
We consider only closed formulas, i.e., formulas where for each atom $a_\pi$ the trace variable $\pi$ is bound by some trace quantifier. 
The semantics of HyperLTL is given with respect to a set of traces $\traceSet$ and a trace assignment $\Pi$, which is a partial mapping $\Pi : \traceVars \rightharpoonup \Sigma^\omega$. 
For $\pi \in \traceVars$ and $t \in T$, we write $\Pi[\pi \mapsto t]$ for the trace assignment obtained by updating the value of $\pi$ to $t$. 
We write $\Pi[i,\infty]$ for the assignment $\Pi[i,\infty](\pi) \coloneqq  \Pi(\pi)[i,\infty]$.
\begin{align*}
   \Pi &\ltlmodels_\traceSet  a_\pi &\text{iff} \quad  &a \in \Pi(\pi)(0)\\
   \Pi &\ltlmodels_\traceSet  \neg \phi &\text{iff} \quad & \Pi \not\ltlmodels_\traceSet  \phi \\
   \Pi &\ltlmodels_\traceSet  \phi_1 \land \phi_2 &\text{iff} \quad  &\Pi \models_\traceSet \phi_1 \text{ and }  \Pi \models_\traceSet  \phi_2\\
   \Pi &\ltlmodels_\traceSet  \ltlN  \phi &\text{iff} \quad & \Pi[1,\infty] \ltlmodels_\traceSet \phi \\
   \Pi &\ltlmodels_\traceSet  \phi_1 \ltlU \phi_2 &\text{iff} \quad & \exists i \ldot \Pi[i, \infty]\ltlmodels_\traceSet  \phi_2 \text{ and } \\
    & & & \quad \forall j < i \ldot  \Pi[j,\infty] \ltlmodels_\traceSet  \phi_1\\
   \Pi &\models_\traceSet \exists \pi \ldot \varphi &\text{iff} \quad &\exists t \in \traceSet \ldot \Pi[\pi \mapsto t] \models_\traceSet  \varphi\\
   \Pi &\models_\traceSet  \forall \pi \ldot \varphi &\text{iff} \quad &\forall t \in \traceSet \ldot \Pi[\pi \mapsto t] \models_\traceSet  \varphi
\end{align*}%
We say that $\traceSet$ is a model of $\varphi$ (written $\traceSet \models \varphi$) if $\emptyset \models_\traceSet \varphi$, where $\emptyset$ denotes the empty trace assignment. 

\begin{remark}
	HyperLTL is closed under conjunction (and, more generally, under any boolean combination). 
	For two HyperLTL formulas $\varphi_1, \varphi_2$, we therefore write $\varphi_1 \land \varphi_2$ for \emph{some} HyperLTL formula expressing the conjunction of $\varphi_1, \varphi_2$.
	For examples $(\forall \pi \exists \pi'\ldot\ltlG( a_\pi \not\leftrightarrow a_{\pi'})) \land (\forall \pi\ldot \ltlF b_\pi)$ can be expressed as $\forall \pi \forall \pi'  \exists \pi''\ldot \ltlG(a_\pi \not\leftrightarrow a_{\pi''}) \land  \ltlF b_{\pi'}$.
	\demo
\end{remark}

Analogous to trace properties, we can characterize hyperproperties as hypersafety and hyperliveness \cite{ClarksonS08}.
We lift the prefix relation $\lessdot$ to sets of traces:
a set $\finTraceSet \subseteq \Sigma^*$ of finite traces is a prefix of a set $\traceSet \subseteq \Sigma^\omega$ (written $\finTraceSet \lessdot \traceSet$) if, for every $u \in \finTraceSet$, there exists a $t \in \traceSet$ such that $u \lessdot t$.

\begin{definition}\label{def:hyperSafe}
    A hyperproperty $H$ is \emph{hypersafety} if for every $\traceSet \subseteq \Sigma^\omega$ with $\traceSet \notin H$, there exists a finite set $\finTraceSet \subseteq \Sigma^*$ with $\finTraceSet \lessdot \traceSet$ such that, for every $\traceSet' \subseteq \Sigma^\omega$ with $\finTraceSet \lessdot \traceSet'$, we have $\traceSet' \notin H$.
    A property $H$ is \emph{hyperliveness} if for every finite set $\finTraceSet \subseteq \Sigma^*$, there exists $\traceSet \subseteq \Sigma^\omega$ with $\finTraceSet \lessdot \traceSet$ and $\traceSet \in H$.
\end{definition}

Intuitively, a violation of a hypersafety property can be explained by the finite interaction of finitely many traces.
Conversely, a hyperproperty is hyperliveness, if such a set can always be extended to a set satisfying the property.

\subsection{Specifications and Notation}
\label{subsec:notation}
We study the combination of $\forall^*\exists^*$ HyperLTL formulas and arbitrary LTL formulas,
and call such pairs \emph{specifications}.

\begin{definition}
    A \emph{specification} is a pair $(\psi, \varphi)$ where $\psi$ is an LTL formula and $\varphi$ a HyperLTL formula. 
    We say that $(\psi, \varphi)$ is satisfiable iff there exists a non-empty set of traces $\traceSet \subseteq \Sigma^\omega$ such that $\forall t \in \traceSet \ldot t \ltlmodels \psi$ and $\traceSet \models \varphi$. 
\end{definition}

In general, we write $(\psi, \varphi)$ for specifications with arbitrary LTL and HyperLTL formulas. 
We use the following notation for fragments of specifications. 
We write $(\top, \varphi)$ to indicate that no LTL formula is given or, equivalently, the trace specification is $\mathit{true}$.
We represent the quantifier prefix of the HyperLTL property using regular expressions.
For example, $\forall\exists^*$ is a prefix consisting of a single universal quantifier followed by any number of existential quantifiers. 
We write $Q^*$ for an arbitrary prefix. 
The body of a HyperLTL formula is structured based on the use of temporal operators.
We allow propositional (temporal-operator-free) formulas as conjuncts if not stated otherwise
Consider the following example.
A $\forall^*\exists^*\ldot\ltlG(\ltlN^*)$ formula is of the form
$\forall \pi_1 \ldots \forall \pi_n \exists \pi_{n+1} \ldots \exists \pi_{n+m} \ldot (\ltlG \phi) \land \phi'$, where $\phi$ may contain (potentially nested) $\ltlN$ operators and $\phi'$ does not contain any temporal operators. 
Analogously, a formula in $\forall^*\exists^*\ldot\ltlG$ describes formulas as the one above but $\phi$ may not contain $\ltlN$s. A formula in $\forall^*\exists^*\ldot \ltlF \land \ltlF$ uses a conjunction of two eventualities (also without $\ltlN$s).

\subsection{Complexity of Undecidable Problems}

Many problems considered in this paper are highly undecidable. 
To enable precise quantification of ``how undecidable'', we briefly recall the arithmetic and analytical hierarchy.
We only provide a brief overview and refer to \cite{rogers1987theory} for details. 
The arithmetic hierarchy contains all problems (languages) that can be expressed in first-order arithmetic over the natural numbers.
It contains the class of recursively enumerable (\texttt{RE}) and co-enumerable problems (\coRE) in its first level. 
The class $\Sigma_1^1$ (sitting in the analytical hierarchy) contains all problems that can be expressed with existential second-order quantification (over sets of numbers) followed by a first-order arithmetic formula.
Analogously, the class $\Pi_1^1$ contains all problems expressible using universal second-order quantification.
Consequently, both $\Sigma_1^1$ and $\Pi_1^1$ (strictly) contain the entire arithmetic hierarchy. 

\subsection{Machines}

As a basic model of computation to show hardness we use two-counter machines. 
A \emph{nondeterministic 2-counter machine} (2CM) consists of a finite set of instructions $l_1, \ldots l_n$, which modify two counters $c_1, c_2$. Each instruction $l_i$ is of one of the following forms, where $x \in \set{1,2}$ and  $1 \leq i,j,k \leq n$.\\
\texttt{$1)$\,$l_i$: $c_x \coloneqq c_x+1$; goto $\{l_j, l_{k}\}$ } \\
\texttt{$2)$\,$l_i$: $c_x \coloneqq c_x-1$; goto $\{l_j, l_{k}\}$ } \\
\texttt{$3)$\,$l_i$: if $c_x = 0$ then goto $l_j$ else goto $l_k$} \\
\texttt{$4)$\,$l_i$: halt}\\
Here, \texttt{goto} $\{l_j, l_{k}\}$ indicates that the machine nondeterministically chooses between instructions $l_j$ and $l_k$.
A configuration of a 2CM is a tuple $(l_i, v_1, v_2)$, where $l_i$ is the next instruction to be executed,
and $v_1, v_2 \in \nat$ denote the values of the counters.
The initial configuration of a 2CM is $(l_1, 0, 0)$. The transition relation between configurations is defined as expected. 
Decrementing a counter that is already $0$ leaves the counter unchanged. 
A 2CM halts if a configuration with a \texttt{halt} instruction is reached.
Deciding if a machine has a halting computation is \texttt{RE}-complete and deciding if it has an infinite computation is \texttt{coRE}-complete~\cite{minsky1967computation}.
An infinite computation is \emph{recurring}  if it visits instruction $l_1$ infinitely many times.
Deciding if a machine has a recurring computation, is $\Sigma_1^1$-hard~\cite{FischerL79,AlurH94}. 

%% ================================
\section{Temporal Safety}\label{sec:safety}
%% ================================
In this section, we study the satisfiability problem of temporal safety HyperLTL formulas.
We begin by defining temporal safety and argue why, compared to hypersafety, it is the more suitable fragment in the context of satisfiability.
Subsequently, we show that temporal safety specifications improve the general $\Sigma_1^1$-hardness of $\forall^*\exists^*$  hyperproperties \cite{FortinKT021} to \coRE-complete.
We obtain this result by a reduction to satisfiability of first-order logic.
In the next step, we investigate the combination of temporally safe hyperproperties with arbitrary functional trace specifications given in LTL. The complexity jumps again to $\Sigma_1^1$-completeness, perhaps surprisingly already for very basic $\forall \exists^*$ formulas only using one $\LTLglobally$ as temporal operator.
We, therefore, analyze the remaining fragments for decidability results and establish that hyperproperties that only use $\LTLnext$s as temporal operators are \NEXPT-complete, even when adding arbitrary LTL specifications. The same holds for hyper-invariants (using $\LTLglobally$) without an LTL specification.

\subsection{Hypersafety and Temporal Safety}
The safety fragment of LTL is one of the most successful fragments of temporal logics as it is amendable to easier monitoring and verification than arbitrary $\omega$-regular properties~\cite{KupfermanV99}.
The concept of a safety property (i.e., every violation is caused after a finite time) naturally extends to hyperproperties, giving the general class of hypersafety (cf.~\Cref{def:hyperSafe}) \cite{ClarksonS08}. 
However, hypersafety is not well suited for a systematic study of the decidability of hyperproperties.
Deciding if a property is hypersafety is already highly undecidable and deciding if a hypersafety property is satisfiable is directly reducible to LTL satisfiability.

\begin{proposition}
    Deciding if a HyperLTL formula $\varphi$ is hypersafety is $\Pi_1^1$-hard. 
\end{proposition}
\begin{proof}
	As shown in \cite[Thm.~23]{DBLP:conf/csfw/FinkbeinerHT19}, for any HyperLTL formula $\varphi$ we can effectively construct a formula $\varphi'$ such that $\varphi$ is unsatisfied iff $\varphi'$ is hypersafety.  
	As HyperLTL unsatisfiability is $\Pi_1^1$-hard~\cite{FortinKT021}, the hardness follows.
\end{proof}

\begin{proposition}
	Given a HyperLTL formula $\varphi$ that is hypersafety, satisfiability of $\varphi$ is decidable in \texttt{PSPACE}.
\end{proposition}
\begin{proof}
    As hypersafety properties are closed under subsets \cite{ClarksonS08}, $\varphi$ is satisfiable iff it is satisfiable by a single trace model. 
    Therefore, we can collapse all quantifiers in $\varphi$ to universal ones, giving an equisatisfiable (but not equivalent) $\forall^*$ formula for which satisfiability is decidable in \texttt{PSPACE} \cite{decidable-hyperltl-1}.
\end{proof}

Instead of focusing on hypersafety, we study the satisfiability problem for a broader fragment of formulas which we call temporally safe.

\begin{definition}
    \label{definition:th}
    A HyperLTL formula $Q \pi_1 \ldots Q \pi_n \ldot \phi$ is \emph{temporal safety} if $\phi$ (interpreted as an LTL formula over $\ap_{\pi_1} \cup \ldots \cup \ap_{\pi_n}$) describes a safety property.
\end{definition}

Similar to the case of LTL \cite{KupfermanV99}, the safety restriction on the body of the HyperLTL formula allows for easier \emph{verification} (see, e.g., \cite{BeutnerF22CAV,BeutnerF22CSF}).
We argue that temporal safety is also an interesting fragment to study in the context of \emph{satisfiability}. 
First, compared with hypersafety, it is decidable whether a formula is temporally safe, as safety is recognizable for LTL~\cite{sistla1994safety}. 
Second, the next two propositions show that temporal safety defines an expressive fragment: it subsumes all $\forall^*\exists^*$ hypersafety properties.

\begin{proposition}\label{prop:hypersafetyVTemSafetyForall}
	For any~$\,\forall^*$ hypersafety property, there exists an equivalent $\forall^*$ property that is temporally safe. 
\end{proposition}
\begin{proof}
	Let $\varphi = \forall \pi_1 \ldots \pi_n \ldot \phi$ be the hypersafety property.
	For any function $f : \{1, \ldots, n\} \to \{1, \ldots, n\}$ (of which there are $n^n$ many) we define the formula $\phi_{[f]}$ as the formula obtained by replacing each trace variable $\pi_i$ for $1 \leq i \leq m$ with $\pi_{f(i)}$.
	Define $\varphi' \coloneqq  \forall \pi_1 \ldots \pi_n \ldot \phi'$ where 
	{\begin{align*}
			\phi' \coloneqq \bigwedge_{f : \{1, \ldots, n\} \to \{1, \ldots, n\}} \phi_{[f]}
	\end{align*}}%
	It is easy to see that $\varphi \equiv \varphi'$ (using the semantics of universal quantification).
	We claim that $\phi'$ expresses a safety property when interpreted as trace property over $\ap_{\pi_1} \cup \cdots \cup \ap_{\pi_n}$.
	 Take any trace $t$ over $\ap_{\pi_1} \cup \cdots \cup \ap_{\pi_n}$ with $t \not\models \phi'$ (as in the definition of safety, cf.~\Cref{def:safetyLiveness}).
	Let $T = \{t_1, \ldots, t_n\}$ be the set obtained by splitting $t$ into $n$ traces, i.e., $t_i$ is a trace over $\ap$ that copies the assignments to $\ap_{\pi_i}$ on $t$.
	By construction of $T$ we get $T \not\models \varphi'$ and, as $\varphi \equiv \varphi'$ is hypersafety, we get a finite set of finite traces $U \lessdot T$ such that no extension of $U$ satisfies $\varphi$.
	We assume that $U = \{u_1, \ldots, u_n\}$ where $u_i \lessdot t_i$ for each $i$.
	This assumption is w.l.o.g., as we can replace multiple prefixes of the same $t_i$ by the longest among those prefixes, and add an arbitrary prefix of each $t_i$ that previously had no prefix in $U$. 
	We further assume, again w.l.o.g., that all $u_i$s have the same length, say $k$.
	Now define $u$ as the finite trace (of length $k$) over $\ap_{\pi_1} \cup \cdots \cup \ap_{\pi_n}$, where the assignment to $\ap_{\pi_i}$ is taken from $u_i$.
	As $u_i \lessdot t_i$ for each $i$, we get $u \lessdot t$.
	It remains to argue that $u$ is a bad prefix of $\phi'$.
	Let $t'$ be any trace with $u \lessdot t'$. 
	We, again, split $t'$ into traces $t'_1, \ldots, t'_n$.
	Now $T' \coloneqq  \{t'_1, \ldots, t'_n\}$ satisfies $U \lessdot T'$, so $T' \not\models \varphi$.
	By the semantics of universal quantification, there thus exists a $f$ such that $[\pi_1 \mapsto t'_{f(1)}, \ldots, \pi_n \mapsto t'_{f(n)}] \not\models \phi$ and so $[\pi_1 \mapsto t'_1, \ldots, \pi_n \mapsto t'_n] \not\models \phi_{[f]}$.
	This implies that $t' \not\models \phi_{[f]}$ in the LTL semantics so $t' \not\models \phi'$ as required.
\end{proof}

\begin{remark}
	We do not claim that every $\forall^*$ hypersafety property is temporally safe.
	Instead,  \Cref{prop:hypersafetyVTemSafetyForall}  only states that there exists an equivalent temporally safe property.
	For example, $\forall \pi \forall \pi'\ldot \ltlF (a_\pi \land \neg a_{\pi'})$ is unsatisfiable and thus hypersafety but $\ltlF (a_\pi \land \neg a_{\pi'})$ is not a safety property.
	\demo
\end{remark}

\begin{proposition}\label{prop:hypersafetyVTemSafety}
    For any $\forall^*\exists^*$ hypersafety property, there exists an equivalent $\forall^*$ property that is temporally safe. 
\end{proposition}
\begin{proof}
    Let $\varphi = \forall \pi_1 \ldots \pi_n \exists \pi'_1 \ldots \pi'_m \ldot \phi$ be hypersafety. 
    For a function $g : \{1, \ldots, m\} \to \{1, \ldots, n\}$ we define the formula $\phi_{[g]}$ as the formula obtained by replacing each trace variable $\pi'_i$ for $1 \leq i \leq m$ with $\pi_{g(i)}$.
    Now define:
    {\begin{align*}
            \varphi' \coloneqq  \forall \pi_1 \ldots \pi_n \ldot \bigvee_{g : \{1, \ldots, m\} \to \{1, \ldots, n\}} \phi_{[g]}
    \end{align*}}%
    We claim that $\varphi \equiv \varphi'$.
    Showing that $\varphi'$ implies $\varphi$ is easy as the disjunction gives an explicit witness for the existential quantifiers.
    For the other direction, assume $\traceSet \models \varphi$ for some model $\traceSet$. 
    Let $t_1, \ldots, t_n \in \traceSet$ be arbitrary. 
    As $\varphi$ is a hypersafety property and $\{t_1, \ldots, t_n\} \subseteq \traceSet$, we get that $\{t_1, \ldots, t_n\} \models \varphi$.
    In particular, if we bind each $\pi_i$ to $t_i$ (in $\varphi$), we get witness traces $t'_1, \ldots, t'_m \in \{t_1, \ldots, t_n\}$ for the existential quantifiers in $\varphi$.
    Now define $g$ by mapping each $1 \leq j \leq m$ to $i \in \{1, \ldots, n\}$ with $t'_j = t_i$.
    The trace assignment $[\pi_1 \mapsto t_1, \ldots, \pi_n \mapsto  t_n]$ satisfies $\phi_{[g]}$.
    As we can find such a $g$ for every $t_1, \ldots, t_n \in \traceSet$, we get that $\traceSet \models \varphi'$ as required. 
    As $\varphi'$ is a $\forall^*$ formula, we can conclude using \Cref{prop:hypersafetyVTemSafetyForall}.
\end{proof}

While temporal safety subsumes $\forall^*\exists^*$ hypersafety, it is a strictly larger fragment as shown by the following formula.
\begin{align*}
    &\big(\exists \pi \ldot a_\pi \big)\land  \big(\forall \pi \ldot \LTLglobally (a_\pi \rightarrow \LTLnext \LTLglobally \neg a_\pi)\big) \land \\
    &\big(\forall \pi \exists \pi' \ldot	\LTLglobally (a_{\pi} \leftrightarrow \LTLnext a_{\pi'})\big)
\end{align*}
Every model of this property must contain \emph{infinitely} many traces.
Hypersafety properties, on the other hand, are closed under subsets and are therefore always satisfiable by a single trace model (if satisfiable at all)~\cite{ClarksonS08}.

%% --------------------------------------
\subsection{Temporal Safety without Functional Specifications}
%% --------------------------------------

Having established that temporal safety spans a broad spectrum of properties, we now establish that the general analytical hardness of HyperLTL~\cite{FortinKT021} drops to \coRE-completeness for temporal safety.  

\begin{theorem}\label{theo:coReComp}
	The satisfiability problem of temporally safe HyperLTL is \coRE-complete.
\end{theorem}

We show the upper bound of \Cref{theo:coReComp} by giving an effective translation from temporally safe HyperLTL to first-order logic using the fact that satisfiability of first-order logic is \coRE-complete \cite{godel1929vollstandigkeit}. 
Our translation enables the application of first-order satisfiability solvers in the realm of hyperproperties.  

\begin{definition}
	A safety automaton over alphabet $\Sigma$ is a tuple $\mathcal{A} = (Q, q_0, \delta)$ where $Q$ is a finite set of states, $q_0 \in Q$ the initial state, and $\delta \subseteq Q \times \Sigma \times Q$ is the transition relation.
	A trace $t \in \Sigma^\omega$ is accepted by $\mathcal{A}$ if there exists \emph{some} infinite run $r \in Q^\omega$ such that $r(0) = q_0$ and for all $i$, $(r(i), t(i), r(i+1)) \in \delta$.
	For every safety trace property $\phi$, there exists a safety automaton that accepts $\phi$ \cite{KupfermanV99}.
\end{definition}

\begin{restatable}{proposition}{coREContainment}
	The satisfiability problem of temporally safe HyperLTL is in \coRE.
	\label{lem:coReCont}
\end{restatable}
\begin{proof}
	Let $\varphi = Q_1 \pi_1 \ldots Q_n \pi_n \ldot \phi$ be a temporally safe HyperLTL formula. 
	Let $\mathcal{A}_\phi = (Q_\phi, q_{0, \phi}, \delta_\phi)$ be a safety automaton over $\Sigma \coloneqq  \pow{\ap_{\pi_1} \cup \cdots \cup \ap_{\pi_n}}$ that accepts $\phi$ (when interpreted as an LTL formula over $\ap_{\pi_1} \cup \cdots \cup \ap_{\pi_n}$).
	We define in the following an equisatisfiable first-order formula $\Theta$, which can be computed from $\varphi$.
	For readability, we use two-sorted first-order logic, which is equisatisfiable to standard first-order logic. 
	We use two sorts: $\mathit{Trace}$, which contains trace variables, and $\mathit{TimePoint}$, which contains time variables. 
	We use the constant $i_0 : \mathit{TimePoint}$ to indicate the initial time point.
	The predicate $\mathit{Succ}(\cdot, \cdot)$ over $\mathit{TimePoint} \times \mathit{TimePoint}$ encodes the successor relation on time.
	For each $a \in \ap$, we use a predicate $P_a(\cdot, \cdot)$ over $\mathit{Trace} \times \mathit{TimePoint}$ to indicate that on trace $t$, $a$ holds at point in time $i$.
	For each state $q \in Q_\phi$ we use a predicate $\mathit{State}_q$ over $\mathit{Trace}^n \times \mathit{TimePoint}$.
	Informally, $\mathit{State}_q(x_1, \ldots, x_n, i)$ indicates that a run of $\mathcal{A}$ on traces $x_1, \ldots, x_n$ is in state $q$ at timepoint $i$.
	
	We first ensure that each point in time has a successor and that the set of traces is non-empty.
	\begin{alignat*}{2}
		& \phi_\text{succ} && \coloneqq \forall i:\mathit{TimePoint} \ldot \exists i' : \mathit{TimePoint} \ldot \mathit{Succ}(i, i') \\
		& \phi_\text{non-empty} && \coloneqq\exists x : \mathit{Trace} \ldot \top
	\end{alignat*}
	For each state $q \in Q_\phi$, we construct a formula $\rho_q$ (over free variables $x_1, \ldots, x_n$), describing that, for any choice of traces and at any point in time, there is a transition in $\mathcal{A}$.
	\begin{align*}
		\rho_q \coloneqq~ & \forall i, i' : \mathit{TimePoint} \ldot  \mathit{State}_q(x_1, \ldots , x_n, i) \land \mathit{Succ}(i, i') \\
		&\rightarrow \bigvee_{(q, \sigma, q')\in \delta_\phi} \Big(
		\bigwedge_{a_{\pi_j} \in \sigma} P_a(x_j, i) \land \bigwedge_{a_{\pi_j} \not\in \sigma} \neg P_a(x_j, i) \, \land \\
		&\quad\quad\quad\quad\quad\quad 	\mathit{State}_{q'}(x_1, \ldots , x_n, i')\Big)
	\end{align*}
	Now, $\Theta$ is defined as follows
	\begin{align*}
		\Theta \coloneqq &~ Q_1 x_1 : \mathit{Trace}\ldot \ldots Q_n x_n : \mathit{Trace} \ldot \, \phi_\text{succ} \land \phi_\text{non-empty}\\ 
		& \quad \land \bigwedge_{q \in Q} \rho_q \land \mathit{State}_{q_0}(x_1, \ldots, x_n, i_o).
	\end{align*}
	The last conjunct ensures that all trace tuples chosen by the quantifiers have an infinite run in $\mathcal{A}$ starting in the initial state and in the initial time point.
    If $\Theta$ is satisfiable, we can construct a trace assignment by setting the propositions based on the evaluation of $P_a(\cdot, \cdot)$ in a satisfying first-order model of $\Theta$ and vice versa.
    A detailed proof can be found in \ifFull{\Cref{app:secSafety}}.
\end{proof}

To complement the upper bound, we show \coRE-hardness by reducing the complement of the halting problem of deterministic Turing machines. 
The proof shows that already a \emph{single} $\ltlG$ with nested $\ltlN$ suffices for \coRE-hardness. 

\begin{restatable}{lemma}{coReHardness}
    The satisfiability problem is \coRE-hard for specifications $(\top, \varphi)$ where $\varphi$ is of the form $ \forall \exists ^* \ldot \LTLglobally (\ltlN^*)$.\label{lem:coReHardness}
\end{restatable}
\begin{proofSketch}
    We encode the non-termination of deterministic Turning machines, which is \coRE-hard.
    Each trace represents a configuration of the machine and the $\forall\exists$ formula demands that each configuration encoded in trace $\pi$ has a successor configuration on some trace $\pi'$.
    As the transitions of a TM can be checked locally, we can encode a successor configuration by comparing every three consecutive positions on $\pi$ and  $\pi'$, which is possible with a single globally. 
    We give a detailed proof in \ifFull{\Cref{app:secSafety}}.%
\end{proofSketch}

This completes the proof of \Cref{theo:coReComp}.

\subsection{Temporal Safety with Functional Specifications}

We now investigate satisfiability for the combination of temporally safe HyperLTL formulas and LTL properties.
If the LTL specification is safety, we can simply merge the trace property with the temporally safe hyperproperty, maintaining the applicability of \Cref{theo:coReComp}.
The situation changes if we allow non-safety trace properties.

\begin{restatable}{theorem}{safetySigma}\label{theo:safetyUndec}
    The satisfiability problem is $\Sigma^1_1$-complete for specifications $(\psi, \varphi)$ where $\varphi$ is of the form $\forall \exists^* \ldot \LTLglobally$.
\end{restatable}
\begin{proofSketch}
    Membership in $\Sigma^1_1$ follows from~\cite{FortinKT021}.
    For hardness, we encode recurring computations of nondeterministic two-counter machines.
    We represent each configuration by encoding the current instruction and two atomic propositions $\mathbf{c}_1, \mathbf{c}_2$ that hold exactly once, i.e., counter $x$ has value $v$ if $\mathbf{c}_x$ occurs in the $v$th position. 
    To ensure a recurring computation, we add a third counter $t$ that decreases in each step. When it reaches $0$, the trace must encode the initial instruction, at which point $t$ is reset to any value. 
    The key idea of the proof is that each trace in the model represents two consecutive configurations, which are encoded over disjoint copies of $\ap$ (for $i \in \{1, 2\}$, $\ap^i = \{a^i \mid a \in \ap\}$).
    In LTL, we can express that the second configuration encoded in a trace is a successor of the first configuration in that trace.
    Furthermore, we express in LTL that the value of $t$ either decreases or the initial instruction is executed. 
    In the HyperLTL property we ensure the existence of the initial configuration, and state that for each trace $\pi$, there exist a $\pi'$ such that the second configuration on $\pi$ equals the first on $\pi'$.
    We can express the latter as
    \begin{align*}
        \textstyle \forall \pi \exists \pi' \ldot \LTLglobally \bigwedge_{a \in \ap} a^2_{\pi} \leftrightarrow a^1_{\pi'}.
    \end{align*}
    The resulting specification is satisfiable if and only if the machine has a recurring computation.
    We give a detailed proof in \ifFull{\Cref{app:secSafety}}. 
\end{proofSketch}

%% --------------------------------------
\subsection{Propositional Hyperproperties and Invariants}
%% --------------------------------------

As we have seen, with arbitrary LTL properties present, a single $\ltlG$ operator suffices to jump to $\Sigma^1_1$.
This leaves hyperproperties expressed using only $\ltlN$s as the only sub-analytical fragment. 
We settle the precise complexity of the resulting problem to be \NEXPT-complete.

\begin{restatable}{theorem}{propTheo}\label{theo:prop}
    The satisfiability problem is \NEXPT-com\-plete for specifications $(\psi, \varphi)$ where $\varphi$ is of the form $Q^* \ldot \ltlN^*$. Hardness holds already for $\psi = \top$, a $\forall^*\exists^*$ prefix, and no $\ltlN$s.
\end{restatable}
\begin{proofSketch}
    To show containment, we nondeterministically guess a set of finite traces $M \subseteq \Sigma^k$, where $k$ is the number of $\ltlN$ operators in the formula. 
    We then verify that each trace in $M$ can be extended to one satisfying $\psi$ and that $M$ is a model of the hyperproperty.
    For the lower bound, we reduce the acceptance of an exponential-time bounded nondeterministic Turing machine to a HyperLTL formula.
    Our encoding is a $\forall^*\exists^*$-formula, which does not contain any temporal operators (not even $\ltlN$) and no trace property. 
    Each trace in our model encodes a piece of information $(s, p, \gamma, q)$, where $s, p \in \nat$, $\gamma$ is a tape symbol of the TM, and $q$ either a state of the TM or $\bot$.
    The tuple $(s, p, t, q)$ encodes that in time-step $s$ and at position $p$, the tape content is $\gamma$, and either the head is at position $p$ and the machine is in state $q$, or the head is not at position $p$ (if $q = \bot$).
    As the TM is time (and thus space) bounded, $s$ and $p$ are bound by $2^n$ for some $n$.
    We show that we can express in HyperLTL that the information encoded in a given model defines a valid accepting run of the TM.
    The resulting formula is satisfiable iff TM has an accepting computation. 
    As we can never refer to all exponentially many positions explicitly, we use $\forall^*\exists^*$ formulas to encode a counter that references all positions. 
    We give a formal proof in the \ifFull{\Cref{app:secSafety}}. 
\end{proofSketch}

We note that HyperLTL without temporal operators has a strong connection to quantified boolean formulas (QBF), the validity of which is a standard \texttt{PSPACE}-complete problem \cite{Stockmeyer76}.
In contrast to QBF, where the quantifier structure spans the polynomial hierarchy, our proof shows that in HyperLTL, the $\forall^*\exists^*$ fragment suffices to show \NEXPT-hardness (refuting a conjecture in \cite{decidable-hyperltl-2} that temporal-operator-free HyperLTL is equivalent to QBF).
The reason is that HyperLTL satisfiability asks for the \emph{existence} of some model for which the formula holds (which is related to the more general second-order QBF problem \cite{Luck16}).

If we forgo the additional trace property, we can also show the following lemma.
Hardness already holds if we disallow propositional formulas outside of the $\ltlG$.  

\begin{lemma}\label{lem:saftyDecNexptime}
    The satisfiability problem is \NEXPT-com\-plete for specifications $(\top, \varphi)$ where $\varphi$ is of the form $\forall^* \exists^* \ldot \ltlG$.
\end{lemma}
\begin{proof}
	A property $\varphi = \forall^n\exists ^m\ldot (\ltlG \phi )\land \phi'$ (where $\phi, \phi'$ do not contain any temporal operators) is satisfiable iff $\forall^n \exists^m \ldot \phi \land \phi'$ is satisfiable.
    The result then follows using \Cref{theo:prop}.
\end{proof}

\section{Temporal Liveness}\label{sec:livness}

The natural counterparts of safety properties are liveness properties, which postulate that ``something good happens eventually''. 
Similar to the case of hypersafety, hyperliveness as a fragment is not well-suited when studying satisfiability: any hyperliveness property is, by definition, satisfiable. 
Analogously to our study of temporal safety, we instead study HyperLTL properties whose body is a liveness property.

\begin{definition}
    \label{definition:tl}
    A HyperLTL formula $Q \pi_1 \ldots Q \pi_n \ldot \phi$ is a \emph{temporal liveness} property if $\phi$ (interpreted as an LTL formula over $\ap_{\pi_1} \cup \cdots \cup \ap_{\pi_n}$) describes a liveness property.
\end{definition}

We examine the temporal liveness fragment following the structure of \Cref{sec:safety} and point out analogous results wherever possible. 
As in \Cref{sec:safety}, we first examine the entire class of temporal liveness and then gradually restrict this class to obtain new decidability results. 

\subsection{Hyperliveness and Temporal Liveness}
As opposed to the safety case (cf.~\Cref{prop:hypersafetyVTemSafety}), temporal liveness and hyperliveness are incomparable.
In temporal liveness, we can easily express falsity via $ \forall \pi \forall \pi'\ldot \ltlF (a_\pi \land \neg a_{\pi'})$, which is not hyperliveness. 
Conversely, the property $\forall\pi \exists \pi'\ldot \ltlG (a_\pi \not\leftrightarrow a_{\pi'})$ is hyperliveness (as we can always add more witness traces) but not expressible in temporal liveness.

\subsection{General Temporal Liveness}

Analogous to \Cref{theo:coReComp}, we consider the full fragment of temporal liveness.
Different from the fragment of temporal safety, the class of temporal liveness is already $\Sigma^1_1$-hard.

\begin{theorem}\label{theo:livenessAnaytical}
    The satisfiability problem is $\Sigma_1^1$-hard for $\forall\exists^* $ temporal liveness HyperLTL formulas.
\end{theorem}

To prove \Cref{theo:livenessAnaytical}, we show a stronger result: we can effectively reduce every $\forall^*\exists^*$ HyperLTL property to an equisatisfiable temporal liveness property.
\begin{theorem}\label{theo:fromGeneralToLiveness}
    Let $(\psi, \varphi)$ be a specification where $\varphi$ is of the form $\forall^n\exists^m \ldot \phi$, and $\psi$ and $\phi$ are arbitrary but satisfiable LTL formulas.
    Then there is an effectively computable specification $(\psi', \varphi')$ such that $\psi'$ is an LTL liveness property, $\varphi'$ is a $\forall^n\exists^m$ temporal liveness property, and $(\psi, \varphi)$ and $(\psi', \varphi')$  are equisatisfiable.
\end{theorem}
\begin{proof}
    The idea is to move the start position of the formula under a $\ltlF$ operator. 
    We introduce a fresh atomic proposition $\dagger$ and ensure that all traces satisfy the liveness property $\ltlF (\dagger \land \ltlN \ltlG \neg \dagger)$. 
    The \emph{unique} position where $\dagger \land \ltlN \ltlG \neg \dagger$ holds (the last time that $\dagger$ is true) is then the ``start position'' to evaluate the formula. 
    Let $\varphi = Q^* \ldot \phi$ where $Q^* = \forall \pi_1 \ldots \pi_n \exists \pi_{n+1} \ldots \pi_{n+m}$ is the quantifier prefix of $\varphi$.
    Define
    \begin{align*}
        \varphi' \coloneqq Q^* \ldot \ltlF \bigg[ \bigwedge_{i=1}^{n+m} \dagger_{\pi_i} \land \Big[\ltlN \ltlG  \bigwedge_{i=1}^{n+m} \neg \dagger_{\pi_i} \Big] \land \phi \bigg]
    \end{align*}
    In similar fashion, we define $ \psi' \coloneqq \ltlF (  \dagger \land (\ltlN \ltlG  \neg \dagger) \land \psi )$.
    It is easy to see that both the LTL body of $\varphi'$ and $\psi'$ are liveness properties.
    Here it is crucial that we assumed that $\psi$ and $\phi$ are satisfiable.
    
    We now claim that $(\psi, \varphi)$ is satisfiable if and only if $(\psi' \varphi')$ is satisfiable.
    For the first direction, assume that $\traceSet$ is a model for $(\psi, \varphi)$.
    The model with $\dagger$ added to the first step of all traces satisfies $(\psi', \varphi')$.
    For the other direction, let $T$ be a model of $(\psi', \varphi')$. We assume w.l.o.g. that there is no subset $T' \subsetneq T$ such that $T'$ is also a model for $(\psi', \varphi')$.
    The property enforces that for any traces $t_1, \ldots, t_{n+m}$, where $t_{n+1}, \ldots, t_{n+m}$ are the witness traces for $t_1, \ldots, t_m$, $\dagger$ holds for the last time at a common time point. 
    As $\traceSet'$ is minimal, every trace serves as a witness for some other traces. Therefore the last position where $\dagger$ holds is the same for all traces in $\traceSet'$.
    Let $i$ be this position.
    Then $\{t[i, \infty] \mid t \in \traceSet'\}$ is a model of $(\psi, \varphi)$.    
\end{proof}

By \Cref{theo:safetyUndec}, satisfiability of $\forall\exists^*$ HyperLTL is $\Sigma^1_1$-hard (note that we transform any specification $(\psi, \forall^n\exists^m \ldot \phi)$ with $n \geq 1$ into a specification $(\top, \forall^n\exists^m \ldot \phi')$ by integrating the trace property into the body of the HyperLTL formula).
\Cref{theo:fromGeneralToLiveness} thus gives a proof of \Cref{theo:livenessAnaytical}.
More generally, \emph{every} HyperLTL formula can be effectively reduced to an equisatisfiable $\forall^2\exists^*$ HyperLTL formula~\cite[Thm.~5]{decidable-hyperltl-2}, so \Cref{theo:fromGeneralToLiveness} shows that deciding temporal liveness can be used to decide full HyperLTL. 

\subsection{Simple Liveness Properties}

The general class of temporal liveness thus does not define an ``easier'' fragment of HyperLTL.
As in the case of safety properties, we study the precise boundary at which the jump to $\Sigma_1^1$ occurs by restricting to simpler forms of temporal liveness. 
Analogously to the case of invariants (described with $\ltlG$), we study eventualities ($\ltlF$).

\begin{restatable}{lemma}{npEventual}
	\label{lem:npEventual}
    The satisfiability problem is \texttt{NP}-complete for specifications $(\top, \varphi)$ where $\varphi$ is of the form $\forall^* \exists^* \ldot \ltlF (\ltlN^*) \land \cdots \land \ltlF (\ltlN^*)$ and no propositional formulas occur outside of the $\ltlF$ operators.
    Hardness already holds for a single eventuality.
\end{restatable}
\begin{proofSketch}
	We collapse all universal quantifiers in $\varphi$ and thereby reduce satisfiability of $(\top, \varphi)$ to boolean satisfiability. 
	We give a detailed proof in \ifFull{\Cref{app:liveness}}.
\end{proofSketch}

Note that if we allow properties where propositional formulas occur outside of the $\ltlF$ operators, the complexity jumps back to \NEXPT (see \ifFull{\Cref{lem:npEventualWithProp}}).
It is worth to contrast this result with the analogous findings for simple temporally safe formulas.
\Cref{lem:npEventual} shows that when adding an $\ltlF$ operator around a propositional formula, the problem drops from \NEXPT (\Cref{theo:prop}) to \NP.
This is in contrast to adding $\ltlG$ operators, which remains \NEXPT-complete (\Cref{lem:saftyDecNexptime}).
Invariants with nested $\ltlN$ and propositional conjuncts are undecidable (\Cref{lem:coReHardness}), whereas eventualities with nested $\ltlN$ operators and propositional conjuncts remain decidable (see \ifFull{\Cref{lem:npEventualWithProp}}).

\subsection{Eventualities with Functional Specifications}

Surprisingly, the sharp contrast between $\ltlG$ and $\ltlF$ continues if we add functional specifications as LTL trace properties. 
For $\ltlG$, the resulting problem directly jumps to full analytical hardness (cf.~\Cref{theo:safetyUndec}). 
For $\ltlF$, we now show that the problem remains decidable. 
Our result reads as follows.

\begin{theorem}\label{theo:eventuallyDec}
    The satisfiability problem is decidable for specifications $(\psi, \varphi)$ where $\varphi$ is of the form $\forall \exists^* \ldot \ltlF(\ltlN^*)$. 
\end{theorem}

This result is interesting for two reasons.
First, it outlines the precise difference between $\ltlG$ and $\ltlF$.
Second, it defines a new decidable class that contains many properties of interest. 
In particular, formulas of the fragment can enforce infinite models.\footnote{Existing decidability results for HyperLTL consider fragments that, if satisfiable, are satisfiable by a finite set of traces of bounded size. This includes the $\exists^*\forall^*$ fragment studied in \cite{decidable-hyperltl-1} and the decidable fragments identified in \cite{decidable-hyperltl-2}.}
For example, the specification in \Cref{ex:introEx} falls in this newly identified fragment.
 
The remainder of this subsection provides a proof for \Cref{theo:eventuallyDec}.
We introduce necessary concepts along the way.   

\subsubsection{Eliminating Nexts} 

We first show how to eliminate the $\ltlN$ operators in $\varphi$.

\begin{restatable}{lemma}{lemmaElimNext}\label{lem:elimNexts}
    Let $(\psi, \varphi)$ be a specification where $\varphi$ is of the form $\forall^n \exists^m \ldot \ltlF (\ltlN^*)$.
	There exists an effectively computable specification $(\psi', \varphi')$ where $\varphi'$ is the of the form $\forall^n \exists^m \ldot \ltlF$ such that $(\psi, \varphi)$ and $(\psi', \varphi')$ are equisatisfiable. 
\end{restatable}
\begin{proofSketch}
    Let $\varphi = \forall^n \exists^m \ldot (\ltlF \phi)\land \phi'$.
    We eliminate $\ltlN$ operators in $\phi$ by letting traces range over tuples. 
    Instead of considering traces in $\Sigma^\omega$, we consider traces in $(\Sigma^k)^\omega$, where $k$ is the lookahead needed to evaluate $\phi$ (which is upper bounded by the number of $\ltlN$s in $\phi$).
    For each trace $t \in \Sigma^\omega$, we define $t' \in (\Sigma^k)^\omega$ by $t'(i) \coloneqq  (t(i), t(i+1), \ldots, t(i+k))$. 
    This reduces the evaluation of $\phi$ to a formula without $\ltlN$s.     
    We also modify the LTL formula (which is allowed to contain $\ltlN$ operators) to assert that the tuples in each tuple trace are consistent, i.e., for each step $i$ if $t(i) = (l_1, \ldots, l_k)$ then $t(i+1) = (l_2, \ldots, l_{k}, l_{k+1})$.
    A detailed proof can be found in \ifFull{\Cref{app:liveness}}. 
\end{proofSketch}

Using \Cref{lem:elimNexts}, we can assume that in Theorem~\ref{theo:eventuallyDec}, the HyperLTL formula $\varphi$ contains a single $\ltlF$ as the only temporal operator. 
For now, we make two further assumptions: 
First, we assume that $\varphi$ contains only a single $\exists$ quantifier, and, second, we assume that there are no additional propositional conjuncts outside the $\ltlF$.
So let $\varphi = \forall \pi\exists \pi' \ldot \ltlF \phi$ where $\phi$ contains no temporal operators.
We begin by translating the trace property $\psi$ into a Büchi automaton.

\begin{definition}
    A \emph{state-labeled Büchi automaton} over alphabet $\Sigma$ is a tuple $\aut = (Q, Q_0,  \delta, F, L)$, where $Q$ is a finite set of states, $Q_0 \subseteq Q$ the initial states, $\delta \subseteq Q \times  Q$ the transition relation, $F \subseteq Q$ the set of accepting states, and $L : Q \to \Sigma$ a state labeling function. 
    An accepting run $r$ of $\aut$ is an infinite sequence $r \in Q^\omega$ such that 1) $r(0) \in Q_0$, 2) $(r(i), r(i+1)) \in \delta$  for every $i$, and 3) $r(i) \in F$ for infinitely many $i$.
    The trace $L(r) \in \Sigma^\omega$ associated to a run is defined by $L(r)(i) \coloneqq  L(r(i))$.
    For a set $X \subseteq Q$, we define $\nstep{\aut}{X} \coloneqq  \{q \in Q \mid \exists q' \in X\ldot (q', q) \in \delta\}$ as all states reachable in one step from $X$ and $\nnstep{\aut}{n}$ as all states reachable in $n$ steps from a state in $Q_0$.
\end{definition}

Note that we use state-labeled automata (as opposed to transition-labeled automata) to simplify our construction.
Let $\aut_\psi =  (Q_\psi, Q_{0, \psi},  \delta_\psi, F_\psi, L_\psi)$ be a (state-labeled) Büchi automaton over $2^\ap$ accepting $\psi$~\cite{VardiW94}.
A state $q \in Q_\psi$ is \emph{non-empty} if there exists an accepting infinite run starting in $q$.
W.l.o.g., we assume that $\aut_\psi$ only includes non-empty states, as we can remove all empty states without changing the language of $\aut_\psi$.
Detecting if a state is non-empty can be done easily using, e.g., nested depth-first search.

\newcommand{\seqSet}{R}

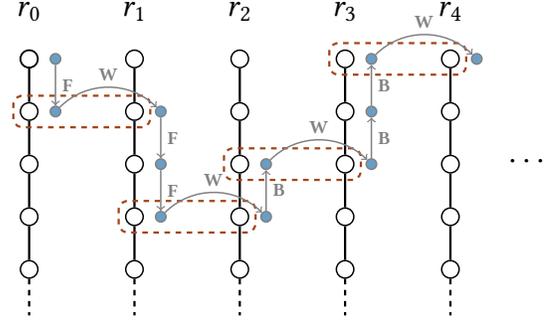
\begin{figure}[t]
	\centering
	\scalebox{0.7}{
		\tikzset{traceNode/.style={circle, draw=black, thick}}
		\tikzset{autNode/.style={circle, fill=pastelblue!70,draw=black!50, thick,inner sep=2pt}}
		\begin{tikzpicture}
			
			\node[] at (0,0.9) () {\huge$r_0$};
			
			\node[traceNode] at (0,0) (n00) {};
			\node[traceNode] at (0,0) (n00) {};
			\node[traceNode] at (0,-1) (n01) {};
			\node[traceNode] at (0,-2) (n02) {};
			\node[traceNode] at (0,-3) (n03) {};
			\node[traceNode] at (0,-4) (n04) {};
			\node[] at (0,-5) (n05) {};
			
			\draw[-, very thick] (n00) -- (n01);
			\draw[-,very thick] (n01) -- (n02);
			\draw[-,very thick] (n02) -- (n03);
			\draw[-, very thick] (n03) -- (n04);
			\draw[-, very thick,dashed] (n04) -- (n05);
			
			\node[] at (2,0.9) () {\huge$r_1$};
			
			\node[traceNode] at (2,0) (n10) {};
			\node[traceNode] at (2,-1) (n11) {};
			\node[traceNode] at (2,-2) (n12) {};
			\node[traceNode] at (2,-3) (n13) {};
			\node[traceNode] at (2,-4) (n14) {};
			\node[] at (2,-5) (n15) {};
			
			\draw[-, very thick] (n10) -- (n11);
			\draw[-,very thick] (n11) -- (n12);
			\draw[-,very thick] (n12) -- (n13);
			\draw[-, very thick] (n13) -- (n14);
			\draw[-, very thick,dashed] (n14) -- (n15);

			\node[] at (4,0.9) () {\huge$r_2$};
			
			\node[traceNode] at (4,0) (n20) {};
			\node[traceNode] at (4,-1) (n21) {};
			\node[traceNode] at (4,-2) (n22) {};
			\node[traceNode] at (4,-3) (n23) {};
			\node[traceNode] at (4,-4) (n24) {};
			\node[] at (4,-5) (n25) {};
			
			\draw[-, very thick] (n20) -- (n21);
			\draw[-,very thick] (n21) -- (n22);
			\draw[-,very thick] (n22) -- (n23);
			\draw[-, very thick] (n23) -- (n24);
			\draw[-, very thick,dashed] (n24) -- (n25);
			
			\node[] at (6,0.9) () {\huge$r_3$};
			
			\node[traceNode] at (6,0) (n30) {};
			\node[traceNode] at (6,-1) (n31) {};
			\node[traceNode] at (6,-2) (n32) {};
			\node[traceNode] at (6,-3) (n33) {};
			\node[traceNode] at (6,-4) (n34) {};
			\node[] at (6,-5) (n35) {};
			
			\draw[-, very thick] (n30) -- (n31);
			\draw[-,very thick] (n31) -- (n32);
			\draw[-,very thick] (n32) -- (n33);
			\draw[-, very thick] (n33) -- (n34);
			\draw[-, very thick,dashed] (n34) -- (n35);
			
			\node[] at (8,0.9) () {\huge$r_4$};
			
			\node[traceNode] at (8,0) (n40) {};
			\node[traceNode] at (8,-1) (n41) {};
			\node[traceNode] at (8,-2) (n42) {};
			\node[traceNode] at (8,-3) (n43) {};
			\node[traceNode] at (8,-4) (n44) {};
			\node[] at (8,-5) (n45) {};
			
			\draw[-, very thick] (n40) -- (n41);
			\draw[-,very thick] (n41) -- (n42);
			\draw[-,very thick] (n42) -- (n43);
			\draw[-, very thick] (n43) -- (n44);
			\draw[-, very thick,dashed] (n44) -- (n45);
			
			\node[] at (9.5,-2) () {\huge$\cdots$};

			\draw[-, very thick, chestnut, dashed, rounded corners=5pt] (1, -0.7) -- (2.3, -0.7) -- (2.3, -1.3) -- (-0.3, -1.3) -- (-0.3, -0.7)  -- (1, -0.7);
			
			\draw[-, very thick, chestnut, dashed, rounded corners=5pt] (3, -2.7) -- (4.3, -2.7) -- (4.3, -3.3) -- (1.7, -3.3) -- (1.7, -2.7)  -- (3, -2.7);
			
			\draw[-, very thick, chestnut, dashed, rounded corners=5pt] (5, -1.7) -- (6.3, -1.7) -- (6.3, -2.3) -- (3.7, -2.3) -- (3.7, -1.7)  -- (5, -1.7);
			
			\draw[-, very thick, chestnut, dashed, rounded corners=5pt] (7, 0.3) -- (8.3, 0.3) -- (8.3, -0.3) -- (5.7, -0.3) -- (5.7, 0.3)  -- (7, 0.3);
			
			\node[autNode] at (0.5,0) (aut0) {};
			\node[autNode] at (0.5,-1) (aut1) {};
			\node[autNode] at (2.5,-1) (aut2) {};
			\node[autNode] at (2.5,-2) (aut3) {};
			\node[autNode] at (2.5,-3) (aut4) {};
			\node[autNode] at (4.5,-3) (aut5) {};
			\node[autNode] at (4.5,-2) (aut6) {};
			\node[autNode] at (6.5,-2) (aut7) {};
			\node[autNode] at (6.5,-1) (aut8) {};
			\node[autNode] at (6.5,0) (aut9) {};
			\node[autNode] at (8.5,0) (aut10) {};
			
			\draw[thick, ->,black!50] (aut0) --node[right] {\textbf{F}} (aut1);
			\draw[thick, ->,black!50] (aut1) to[out=45, in=135] node[above] {\textbf{W}} (aut2);
			\draw[thick, ->,black!50] (aut2) --node[right] {\textbf{F}} (aut3);
			\draw[thick, ->,black!50] (aut3) --node[right] {\textbf{F}} (aut4);
			\draw[thick, ->,black!50] (aut4) to[out=45, in=135] node[above] {\textbf{W}} (aut5);
			
			\draw[thick, ->,black!50] (aut5) --node[right] {\textbf{B}} (aut6);
			\draw[thick, ->,black!50] (aut6) to[out=45, in=135] node[above] {\textbf{W}} (aut7);
			\draw[thick, ->,black!50] (aut7) --node[right] {\textbf{B}} (aut8);
			\draw[thick, ->,black!50] (aut8) --node[right] {\textbf{B}} (aut9);
			\draw[thick, ->,black!50] (aut9) to[out=45, in=135] node[above] {\textbf{W}} (aut10);
			
		\end{tikzpicture}
	}
	
	\caption{Model for $\forall \pi \ldot \exists \pi' \ldot {\protect \LTLeventually} \, \phi$ formulas. Dashed boxes indicate the witness points for the ${\protect \LTLeventually}$ operator. }\label{fig:model}
\end{figure}

\subsubsection{Models for \boldmath$\forall\exists$}
Intuitively, our decidability result can be derived as follows. 
Assume we had a model $\traceSet$ of $(\psi, \varphi)$.
Let $\seqSet \subseteq Q_\psi^\omega$ be a set of accepting runs of $\aut_\psi$ associated to $\traceSet$, i.e., $\traceSet = \{L_\psi(r) \mid r \in \seqSet\}$.
As we consider a $\forall\exists$ formula, we can arrange the runs in $\seqSet$ as a sequence:
we choose $r_0, r_1, \ldots \in \seqSet$ (not necessarily distinct) such that, for each $i$, $r_{i+1}$ serves as a witness for $r_i$, i.e., $[\pi \mapsto L_\psi(r_i), \pi' \mapsto L_\psi(r_{i+1})] \models \ltlF \phi$. 
We say $n_0, n_1, \ldots$ are \emph{witness points} if $[\pi \mapsto L_\psi(r_i), \pi' \mapsto L_\psi(r_{i+1})][n_i, \infty] \models \phi$ for every $i$, i.e., the $n_i$ point to a step at which the eventuality holds.
The trace arrangement is depicted in Figure~\ref{fig:model} (ignoring the blue smaller nodes and gray edges for now). 
For each $i$, the dashed box denotes the witness point $n_i$ where $r_i$ and $r_{i+1}$ satisfy $\phi$.

As an intermediate step, we describe an infinite-state Büchi system (a Büchi automaton without labels) that guesses such a ``linear'' model of $(\psi, \varphi)$.
The states of the system are triples $(q, b, n)$, where $q \in Q_\psi$ is a state in $\aut_\psi$, $b \in \{\forward, \backward\}$ gives a \emph{running direction}, and $n \in \nat$.
Each state $(q, b, n)$ additionally satisfies $q \in \nnstep{\aut_\psi}{n}$.
The initial states of the system are all states $(q_0, \forward, 0)$ with $q_0 \in Q_{0, \psi}$. 
In each step, the system has three options: it can either run forwards, run backwards, or claim to have found a witness.
In a forward step (\textbf{F}-step), the automaton moves from $(q, \forward, n)$ to $(q', \forward, n+1)$, where $(q,q') \in \delta_\psi$.
Similarly, in a backwards step (\textbf{B}-step), it runs from $(q, \backward, n+1)$ to $(q', \backward, n)$, where $(q', q) \in \delta_\psi$.
Note that in the backwards step, we always ensure that $q' \in \nnstep{\aut_\psi}{n}$.
Lastly, the system can claim to have found a witness (\textbf{W}-step): if in state $(q, b, n)$, it can select any $q' \in \nnstep{\aut_\psi}{n}$ such that $L_\psi(q) \times L_\psi(q') \models \phi$. 
Afterwards, the system continues in state $(q', b', n)$, where $b' \in \{\forward, \backward\}$ is chosen nondeterministically.
Call the resulting system $\system$. We claim the following.

\begin{lemma}\label{prop:infSystem}
    $\system$ has an infinite run that uses \textbf{W}-steps infinitely often if and only if $(\psi, \varphi)$ is satisfiable. 
\end{lemma}
\begin{proof}
    We sketch both directions.
    For the ``if'' direction, assume there is a model for $(\psi, \varphi)$.
    We can arrange a subset of this model as depicted in Figure~\ref{fig:model}. 
    Let $r_0, r_1 , \ldots,$ with $r_i \in Q_\psi^\omega$ be the sequence of accepting runs in $\aut_\psi$ and $n_i$ the witness points.
    Traversing the graph as shown by the small blue states in Figure~\ref{fig:model} creates a run of $\system$.
    We start in $(r_0(0), \forward, 0)$ and move forward (using \textbf{F}-steps) until the counter reaches $n_0$.
    At this point, we take the \textbf{W}-step from $(r_0(n_0),\forward, n_0)$ to $(r_1(n_0), b, n_0)$ and run towards counter value $n_1$.
    If $n_1 < n_0$, we set the running direction $b$ to $\backward$ and otherwise to  $\forward$.
    We continue this procedure to construct an infinite run. 
    For the example situation depicted in Figure~\ref{fig:model}, the resulting run would start with:\smallskip
    
    \noindent\scalebox{0.93}{
    \begin{minipage}{\columnwidth}\noindent%
        \begin{align*}
            &(r_0(0), \forward, 0) \xrightarrow{\textbf{F}} (r_0(1), \forward, 1) \xrightarrow{\textbf{W}} (r_1(1), \forward, 1) \xrightarrow{\textbf{F}} (r_1(2), \forward, 2) \\
            &\xrightarrow{\textbf{F}} (r_1(3), \forward, 3) \xrightarrow{\textbf{W}} (r_2(3), \backward, 3) 
            \xrightarrow{\textbf{B}} (r_2(2), \backward, 2) \xrightarrow{\textbf{W}}  \cdots
        \end{align*}
    \end{minipage}
    }\medskip
    
    \noindent The resulting sequence is a run of $\system$ and uses \textbf{W}-steps infinitely many times. 
    
    For the ``only if'' direction, assume an infinite run $r = (q_0, b_0, m_0) \to (q_1, b_1, m_1) \to \cdots$ of $\system$. 
    We split $r$ into infinitely many finite segments $x_0,x_1, \ldots, $ by splitting each time $r$ takes a \textbf{W}-step.
    In the example run above we would get 
    $x_0 = (r_0(0), 0)(r_0(1), 1)$, $x_1 = (r_1(1), 1)(r_1(2), 2) (r_1(3), 3), \ldots$. 
    In general, let $x_i$ be the sequence $(q_i^0, n_i^0)\ldots(q_i^{k_i}, n_i^{k_i})$. 
    From $x_i$, we construct a finite run $r_i \in Q_\psi^*$ of $\aut_\psi$ starting in a state in $Q_{0, \psi}$ such that for every $0 \leq j \leq k_i$, $r_i(n_i^j) = q_i^j$.
    Using the fact that for each $(q, b, n)$ in $\system$, we have $q \in \nnstep{\aut_\pi}{n}$, this is always possible.
    It is crucial that we cannot reverse directions between two \textbf{W}-steps.
    The finite $r_i$ ends in a state in $\aut_\psi$, so by the assumption that all states are non-empty, we can extend it into an infinite accepting run. 
    The set $\set{L_\psi(r_0),L_\psi(r_1), \ldots}$ is a model of $(\psi, \varphi)$.
\end{proof}

\subsubsection{From Infinite State to Pushdown.}
The construction of $\system$ requires infinitely many states as we need to carry the natural number $n$ to ensure valid \textbf{B} and \textbf{W} steps (which need access to $\nnstep{\aut_\psi}{n}$).
We show next that we can replace this infinite state space by a finite pushdown system.

\begin{definition}\label{def:pushdown}
    A \emph{Büchi pushdown system} is a tuple $\pushAut = (Q, \Gamma, Q_0, \gamma_0, \delta, F)$, where $Q$ is a finite set of states, $\Gamma$ the finite stack alphabet, $Q_0 \subseteq Q$ initial states, $\gamma_0 \in \Gamma$ the initial stack symbol, $\delta \subseteq (Q \times \Gamma^+) \times (Q \times \Gamma^*)$ a \emph{finite} transition relation, and $F \subseteq Q$ a set of accepting states. 
    The system operates on configuration $\langle q, \alpha\rangle$, where $q \in Q$ and $\alpha \in \Gamma^*$.
    A transition $\langle q, \alpha\rangle \rightsquigarrow \langle q', \alpha'\rangle \in \delta$ describes that the system, if in state $q$ and $\alpha \in \Gamma^+$ is a prefix of the current stack, pops $\alpha$, pushes $\alpha' \in \Gamma^*$ to the stack and moves to state $q'$. 
    An accepting run is an infinite sequence of configurations that starts in $\langle q_0, [\gamma_0] \rangle$ for some $q_0 \in Q_0$, respects $\delta$, and visits states in $F$ infinitely many times.
    It is decidable in polynomial time if a Büchi pushdown system has an accepting run~\cite{BouajjaniEM97}.
\end{definition}

We replace $\system$ with a pushdown system $\pushAut$.
Conceptually, we represent a state $(q, b, n)$ in $\system$ by the pushdown configuration with state $(q, b)$ and stack content $[\nnstep{\aut_\psi}{n}, \ldots, \allowbreak \nnstep{\aut_\psi}{0}]$, i.e., the length of the stack is $n+1$ and the $i$th element are all states reachable in $i$ steps.
The states in the pushdown system thus have the form $(q, b)$ with $q \in Q_\psi, b \in  \{\forward, \backward\}$ and the stack alphabet is $2^{Q_\psi}$.
The initial stack symbol is  $\gamma_0 \coloneqq  Q_{0, \psi}$ and the initial states are $\{(q_0, \forward) \mid q_0 \in Q_{0, \psi} \}$.
The transitions are of the following form:

\begin{prooftree}
	\def\defaultHypSeparation{\hskip .15in}
	\AxiomC{$(q, q') \in \delta_{\psi}$}
	\LeftLabel{{(\textbf{F})}}
	\UnaryInfC{$\big\langle (q, \forward), [A]\big\rangle \rightsquigarrow \big\langle (q', \forward), [\nstep{\aut_\psi}{A}, A]\big\rangle$}
\end{prooftree}

\begin{prooftree}
	\def\defaultHypSeparation{\hskip .15in}
	\AxiomC{$q' \in A_2$}
	\AxiomC{$(q', q) \in \delta_{\psi}$}
	\LeftLabel{{(\textbf{B})}}
	\BinaryInfC{$\big\langle (q, \backward), [A_1, A_2]\big\rangle \rightsquigarrow \big\langle (q', \backward), [A_2]\big\rangle$ }
\end{prooftree}

\begin{prooftree}
	\def\defaultHypSeparation{\hskip .15in}
	\AxiomC{$q' \in A$}
	\AxiomC{$L_{\psi}(q) \times L_{\psi}(q') \models \phi$}
	\AxiomC{$b, b' \in \{\forward, \backward\}$}
	\LeftLabel{{(\textbf{W})}}
	\TrinaryInfC{$\big\langle (q, b), [A]\big\rangle \rightsquigarrow \big\langle (q', b'), [A]\big\rangle$}
\end{prooftree}

\noindent
Note the close correspondence with the transitions in $\system$.
In particular, in \textbf{F}-steps, we compute $\nnstep{\aut_\psi}{n+1}$ based on $\nnstep{\aut_\psi}{n}$. 
In \textbf{B}-steps, the stack provides access to all states that are reachable, and thus guarantees the invariant that $q \in \nnstep{\aut_\psi}{n}$ for each state $(q, b, n)$ in $\system$.
It is not hard to see that $\pushAut$ has a run that uses \textbf{W}-steps infinitely often iff $\system$ has a run that uses \textbf{W}-steps infinitely often. 
Combined with \Cref{prop:infSystem} we thus get:
\begin{lemma}
    $\pushAut$ has an accepting run that uses \textbf{W}-steps infinitely often if and only if $(\psi, \varphi)$ is satisfiable. 
\end{lemma}
Lastly, we can easily translate a Büchi pushdown system with transition-based acceptance (as in $\pushAut$) to state-based acceptance (as in \Cref{def:pushdown}). Using the decidability of pushdown systems \cite{BouajjaniEM97}, we thus get that the satisfiability of  $(\psi, \varphi)$ is decidable.
Note that our proof gives an elementary upper bound of 2\texttt{EXPTIME} (for $\forall\exists$ properties).\footnote{The size of $Q_\psi$ is at most exponential in $\psi$ \cite{VardiW94}, so the size of the stack alphabet of $\pushAut$ (which is $2^{Q_\psi}$) is at most double exponential in $\psi$. As deciding the emptiness of a Büchi pushdown system is polynomial, the 2\texttt{EXPTIME} upper bound follows.  }

\subsubsection{Propositional Conjuncts and $\forall\exists^*$.}
We can now lift the two assumptions we made earlier.
As a first extension, we modify our construction to also support formulas of the form $\forall\pi \exists \pi'.(\ltlF \phi) \land \phi'$.
To do so, we keep track of the \emph{first} state of the run we are currently considering. 
In a \textbf{W}-step, we then only select a witness state $q'$ that stems from an initial state which satisfies the propositional requirement $\phi'$ when combined with the initial state of the current run.
We can access the set of all such states by keeping track of the set of states reachable from every individual state (by changing the stack alphabet to functions $Q_\psi \to 2^{Q_\psi}$). 

As a second extension, we can show decidability for a $\forall\exists^m$ prefix by moving to \emph{alternating} Büchi pushdown systems (defined as expected, see \cite{BouajjaniEM97} for details). 
For $\forall\exists^m$,  we can no longer arrange the traces of a model in a linear sequence (as depicted in \Cref{fig:model}) and instead use $m$-ary trees labeled by traces such that the children of a node correspond to witness traces of that trace. 
In a \textbf{W}-step from a state $(q, b)$, we now select $m$ states $q_1, \ldots, q_m$ (whereas we previously picked only $q'$) such that $q$ together with $q_1, \ldots, q_m$ satisfy $\phi$.
Afterwards, we need to find a new witnesses for \emph{each} of the $q_i$. 
We accomplish this by introducing a \emph{universal} transition that branches into states $(q_i, b_i)$ for each $1 \leq i \leq m$ (leaving the stack unchanged as before). 
The \textbf{F} and \textbf{B} step stay purely nondeterministic. 
The resulting alternating pushdown system has an accepting run (which now has the form of a tree) iff $(\psi, \varphi)$ is satisfiable. 
As emptiness of alternating pushdown systems is still decidable (albeit only in exponential time)~\cite{BouajjaniEM97}, we get a proof of \Cref{theo:eventuallyDec} for the full $\forall\exists^*$-fragment.
For the $\forall\exists^*$-fragment, our proof gives an elementary upper bound of 3\texttt{EXPTIME}.

\subsection{Conjunctions of Eventualities}

We show that Theorem~\ref{theo:eventuallyDec} is tight in the sense that already a conjunction of eventualities combined with an arbitrary trace property is undecidable (and even $\Sigma_1^1$-hard).

\begin{theorem}\label{theo:conjunctionFAnalytical}
    The satisfiability problem is $\Sigma^1_1$-hard for specifications $(\psi, \varphi)$ where $\varphi$ is of the form $\forall \exists^* \ldot \ltlF \land \ltlF \land \ltlF$. 
\end{theorem}
\begin{proof}
    We encode the problem of whether a nondeterministic 2CM with instructions $l_1, \ldots, l_n$ has a recurring computation \cite{FischerL79,AlurH94}.
    Let $\mathit{AP} = \bigcup_{x \in \{1, 2, t\}} \{\boxdot_x, \blacksquare_x, \mathit{isZero}_x \} \cup \{\mathbf{l}_1, \ldots, \mathbf{l}_n\}$. 
    Each trace encodes a configuration of the machine as follows.
    The current value of counter $x \in \{1, 2\}$ is encoded as a trace in $\emptyset^*\{\boxdot_x\}\{\blacksquare_x\}\emptyset^\omega$ such that the (unique) step at which $\boxdot_x$ holds indicates the current value of $c_x$.
    We later use proposition $\blacksquare_x$ (which always holds the step after $\boxdot_x$) to encode the update of the counter. 
    The proposition $\mathit{isZero}_x$ holds exactly if the counter is zero. 
    The current instruction is encoded by propositions $\{\mathbf{l}_1, \ldots, \mathbf{l}_n\}$, of which exactly one holds globally along a trace. 
    Finally, to ensure a recurring computation, we use a third counter $t$, which is encoded analogously to the counters above and counts down the steps to the next visit to $l_1$. 
    It is easy to see that we can encode the validity of a configuration in an LTL formula $\psi$.
    For $x \in \{1, 2, t\}$ we ensure a valid counter via
    \begin{align*}
        (\neg \boxdot_x \land \neg \blacksquare_x) \ltlU \big(  &(\boxdot_x \land \neg \blacksquare_x) \\
        & \land \ltlN (\blacksquare_x \land \neg \boxdot_x) \land \ltlN\ltlN \ltlG (\neg \boxdot_x \land \neg \blacksquare_x) \big)
    \end{align*}
	and ensure correct placement of $\mathit{isZero}_x$ by  $(\ltlG \mathit{isZero}_x)\lor (\ltlG \neg \mathit{isZero}_x )$ together with $ \mathit{isZero}_x \leftrightarrow \boxdot_x$.
  	Finally, we assert that the propositions $\{\mathbf{l}_1, \ldots, \mathbf{l}_n\}$ are set correctly via $\bigvee_{i} \big( \ltlG \mathbf{l}_i \land \bigwedge_{j \neq i} \ltlG \neg \mathbf{l}_j\big)$.
    In the hyperproperty, we encode that there exists a trace representing the initial configuration as follows (note that we allow counter $t$ to have any value):
    \begin{align*}
        \varphi_\mathit{init} \coloneqq  \exists \pi. (\mathbf{l}_1)_\pi \land (\mathit{isZero}_1)_\pi \land (\mathit{isZero}_2 )_\pi
    \end{align*}
    Lastly, we express that each trace has a successor.
    For each instruction $l_i$, we write $c(l_i) \in \{1, 2\}$ for the counter that is changed or tested in instruction $l_i$. 
    We define $\overline{1} \coloneqq  2$ and $\overline{2} \coloneqq  1$ for the other counter.
    We then define
     {\begin{align*}
            \varphi \coloneqq  \forall \pi\exists \pi'. 
            &\Big[\ltlF \bigvee_{i \in \{1, \ldots, n\}} (\mathbf{l}_i)_{\pi} \land \mathit{exec}(l_i)\Big] \land {}\\
            & \Big[\ltlF \bigvee_{\substack{i \in \{1, \ldots, n\}}} (\mathbf{l}_i)_\pi  \land  (\boxdot_{\overline{c(l_i)}})_{\pi} \land (\boxdot_{\overline{c(l_i)}})_{\pi'}  \Big] \land {}\\
            & \Big[\ltlF  \big((\mathit{isZero}_t )_\pi \land (\mathbf{l}_1)_{\pi}\big) \lor  \big((\boxdot_t)_{\pi} \land (\blacksquare_t)_{\pi'}\big)  \Big].
    \end{align*}}%
    Here, $\mathit{exec}(l_i)$ denotes that the action or test of instruction $l_i$ is performed on $c(l_i)$. 
     For example, if $l_i = \big[c_x \coloneqq c_x+1; \texttt{goto} \{l_j, l_{k}\}\big]$, we define $\mathit{exec}(l_i)$ as
    \begin{align*}
    	\big((\mathbf{l}_j)_{\pi'}  \lor (\mathbf{l}_k)_{\pi'}\big) \land  (\blacksquare_x)_{\pi} \land (\boxdot_x)_{\pi'}.
    \end{align*}%
    Note that $(\blacksquare_x)_{\pi} \land (\boxdot_x)_{\pi'}$ encodes that the counter $x$ is increased.
    For a decrement operation, we can replace this with  $(\blacksquare_x)_{\pi'} \land (\boxdot_x)_{\pi}$.
   If $l_i = \big[\texttt{if } c_x = 0 \texttt{ then goto }  l_j   \texttt{ else} \allowbreak\texttt{goto } l_k\big]$, we define $\mathit{exec}(l_i)$ as
    \begin{align*}
        &(\boxdot_x)_{\pi} \land (\boxdot_x)_{\pi'} \land \big((\mathit{isZero}_x)_{\pi} \rightarrow (\mathbf{l}_j)_{\pi'} \big) \\
        & \land \big(\neg (\mathit{isZero}_x)_{\pi} \rightarrow  (\mathbf{l}_k)_{\pi'} \big).
    \end{align*}%
    In $\varphi$, the first conjunct thus encodes that the counter $c(l_i)$ is updated and/or tested as required by $l_i$.
    The second conjunct states that the counter that is not involved in $l_i$ is left unchanged.
    As the current instruction is set consistently along a trace, both eventualities refer to the same instruction.
    Finally, the third conjunct ensures that the counter $t$ either decreases or is already zero, at which point the current instruction is $l_1$. 
    In case the $t$-counter is zero, it can be reset to any value on $\pi'$. 
    This ensures a recurring computation of the machine. 
    It is easy to see that $(\psi, \varphi_\mathit{init} \land \varphi)$ is satisfiable iff the 2CM has a recurring computation (note that $\varphi_\mathit{init} \land \varphi$ is a $\forall\exists^2$-formula).
\end{proof}

While \Cref{theo:conjunctionFAnalytical} requires three conjunctions of eventualities to show $\Sigma^1_1$-hardness, already two eventualities suffice to show undecidability.
To do so, we can encode the non-termination of a 2CM (avoiding the $t$ counter).
This further underlines the tightness of \Cref{theo:eventuallyDec}.

\begin{lemma}
    The satisfiability problem is undecidable for specifications $(\psi, \varphi)$ where $\varphi$ is of the form $\forall \exists^* \ldot \ltlF \land \ltlF$. 
\end{lemma}

\begin{example}
    Using similar ideas as in \Cref{theo:conjunctionFAnalytical}, we can encode that a one-counter machine has an infinite computation with only a single eventuality (as we only need to ensure that the \emph{single} counter is updated consistently).
    Combined with \Cref{theo:eventuallyDec}, we derive that we can decide the existence of an infinite computation in a one-counter machine. 
    While this is long known (see, e.g., \cite{HaaseKOW09}), it nevertheless emphasizes that our newly identified decidable class is broader than it seems at first glance. 
    \demo
\end{example}

\subsection{Deterministic Liveness}

Livness for trace properties (cf.~\Cref{def:safetyLiveness}) and hyperliveness (cf.~\Cref{def:hyperSafe}) already imply that a property is satisfiable.
As demonstrated by \Cref{theo:livenessAnaytical}, the same does not hold for temporal liveness hyperproperties. 
We can, however, identify a fragment within temporal liveness for which the intuition that liveness implies satisfiability transfers to the realm of hyperproperties. 
We say an LTL property $\phi$  is a \emph{deterministic liveness} property if it is a liveness property and can be recognized by a \emph{deterministic} Büchi automaton. 

\begin{restatable}{proposition}{alwaysSat}
	\label{prop:alwaysSat}
    HyperLTL formulas of the form $\varphi = \forall\exists^* \ldot \phi$ where $\phi$ is a deterministic liveness property are always satisfiable and have a finite model. 
\end{restatable}
\begin{proofSketch}
	In a deterministic automaton describing a liveness property, any reachable state has a path to an accepting state.
	We use this to iteratively construct a model. 
	The full proof can be found in \ifFull{\Cref{app:liveness}}.
\end{proofSketch}

Note that the same does not hold if we consider more than one universal quantifier. 
As an example, the formula $\forall \pi \forall \pi'\ldot \ltlF(a_{\pi} \land \neg a_{\pi'})$ is unsatisfiable but $\ltlF(a_{\pi} \land \neg a_{\pi'})$ is a deterministic liveness property.
If we consider deterministic liveness in combination with trace properties, we again obtain a jump to $\Sigma^1_1$-hardness.

\begin{corollary}
	\label{cor:detLiveness}
    Satisfiability is $\Sigma^1_1$-hard for specifications of the form $(\psi, \varphi)$ where $\varphi$ is of the form $\forall\exists^* \ldot \phi$ and $\phi$ is a deterministic liveness property.
\end{corollary}
\begin{proof}
    Follows directly from \Cref{theo:conjunctionFAnalytical} as conjunctions of eventualities are deterministic liveness.
\end{proof}

\subsection{Overview: Liveness vs Safety}
Our results provide a clear picture of the (un)decidability boundaries within fragments of HyperLTL.
In particular, our systematic study allows a direct comparison between temporal safety and temporal liveness. 
For the full fragment, temporal liveness already subsumes satisfiability of full HyperLTL, which contrasts strongly with the much cheaper (albeit still undecidable) problem for temporal safety. 
This changes if we consider simpler fragments. 
Here, the $\ltlF$ fragment is drastically better behaved in terms of complexity and even admits large decidable fragments for cases where the safety counterpart already exhibits full analytical hardness. 

\section{Finding Largest Models}\label{sec:largestModelsAlg}

To complement the decidability results from the previous sections, we propose a new (incomplete) algorithm to detect (un)satisfiability of $\forall\exists^*$ HyperLTL formulas.
So far, the only available algorithm checks for finite models of bounded size and then iteratively increases the bound~\cite{decidable-hyperltl-2,FinkbeinerHH18}. Such an approach finds smallest models and cannot determine unsatisfiability.
The key insight for our algorithm is that $\forall \exists^*$ formulas are closed under union, therefore, a formula $\varphi$ is satisfiable iff there is a (unique) \emph{largest} model satisfying $\varphi$.
To find such models algorithmically, we iteratively eliminate choices for the $\exists^*$ quantifiers that admit no witness trace when chosen for the $\forall$ quantifier. 
Thereby, we do not only find largest models but can also detect unsatisfiability.
Our incremental elimination is closely related to a recent algorithm used in the context of finite-trace properties (which was developed independently) \cite{bonakdarpour2022finiteword}.

For presentation reasons, we present the algorithm for $\forall\exists$ formulas. 
Our implementation (see \Cref{sec:eval}) supports full $\forall\exists^*$ properties. 

\begin{algorithm}[t]
	\caption{Algorithm that searches for the largest model of a $\forall\exists$ property. Initially, $\mathcal{A}$ is a Büchi automaton that accepts the body the HyperLTL property. }
	\label{alg}
	\begin{algorithmic}[1]
		\Procedure{findModel}{$\mathcal{A}$}
		\If{$\mathcal{L}(\mathcal{A}^\forall) = \emptyset$}
		\State \textbf{return} UNSAT;
		\EndIf	
		\If{$\mathcal{L}(\mathcal{A^\exists}) \subseteq \mathcal{L}(\mathcal{A^\forall})$}
		\State \textbf{return} SAT, model: $\mathcal{L}(\mathcal{A^\forall})$;
		\EndIf	
		\State $\mathcal{A}_\text{new} \coloneqq  \mathcal{A} \cap \mathcal{A}^\forall_{\pi'}$;
		\State \Call{findModel}{$\mathcal{A}_\text{new}$};
		\EndProcedure
	\end{algorithmic}
\end{algorithm}

\subsection{Algorithm}

For a Büchi automaton $\mathcal{A}$ over $\ap_\pi \cup \ap_{\pi'}$, we define $\mathcal{A}^\forall$ and $\mathcal{A}^\exists$ as the automata (over $\ap$) that (existentially) project $\mathcal{A}$ on the alphabet $\ap_\pi$ and $\ap_{\pi'}$, respectively.
Now let a HyperLTL formula $\varphi = \forall \pi\exists \pi' \ldot \phi$ be given and let $\mathcal{A}_\phi$ be an automaton over $\ap_\pi \cup \ap_{\pi'}$ accepting $\phi$.
In particular, $\mathcal{A}_\phi^\forall$ accepts all words for which there exists a witness trace for the existential quantifier. 
Our algorithm is depicted in Algorithm~\ref{alg}. 
Initially, we call \Call{findModel}{$\mathcal{A}_\phi$}.

The first candidate is $\mathcal{A} = \mathcal{A}_\phi$.
If $\mathcal{L}(\mathcal{A}^\forall) = \emptyset$, i.e., no trace has a witness trace in $\mathcal{A}$, $\varphi$ is unsatisfiable. 
If all potential witness traces in $\mathcal{L}(\mathcal{A}^\exists)$ are contained in $\mathcal{L}(\mathcal{A}^\forall)$ (so they have a witness trace themself), $\varphi$ is satisfiable and $\mathcal{L}(\mathcal{A}^\forall)$ is a model.
If neither is the case, we refine $\mathcal{A}$ by removing all traces whose $\exists$ component is not in $\mathcal{L}(\mathcal{A}^\forall)$.
We define $\mathcal{A}_\text{new}$ as the intersection $\mathcal{A} \cap \mathcal{A}^\forall_{\pi'}$ where $\mathcal{A}^\forall_{\pi'}$ is $\mathcal{A}^\forall$ with the alphabet changed from $\ap$ to $\ap_{\pi'}$.
We can compute $\mathcal{A}_\text{new}$ via a standard intersection construction on Büchi automata. 
Note that $\mathcal{A}_\text{new}$ might again contain witness traces that themselves have no witness trace, so we recurse. 

\subsection{Correctness}

The algorithm maintains the following invariants.
\begin{restatable}{lemma}{algInv}\label{prop:algorithm}
    In every iteration of the algorithm it holds that $\lang{\mathcal{A}_\text{new}} \subseteq \lang{\mathcal{A}}$, and for any trace set $\traceSet$ with $\traceSet \models \forall \pi \exists \pi' \ldot \phi$, $\traceSet \subseteq  \mathcal{L}(\mathcal{A^\forall})$.
\end{restatable}

Using \Cref{prop:algorithm} it is easy to see the following.

\begin{proposition}
	Given a formula $ \varphi = \forall \pi \exists \pi' \ldot \phi$, if Algorithm~\ref{alg} terminates with UNSAT, the formula is unsatisfiable. If it terminates with SAT and model $\mathcal{L}(\mathcal{A^\forall})$, then $\mathcal{L}(\mathcal{A^\forall})$ is the unique largest model of $\varphi$.
\end{proposition}

To generalize to $\forall\exists^*$, we intersect the universal projection with each of the projections on existentially quantified positions. 
Models for $\forall^*\exists^*$-properties are, in general, not closed under union, so our algorithm does not extend beyond $\forall\exists^*$.

\section{Implementation and Experiments}\label{sec:eval}

We have implemented the algorithm described in \Cref{sec:largestModelsAlg} in a tool called \texttt{LMHyper} (short for \textbf{L}argest \textbf{M}odel of \textbf{Hyper}LTL).
\texttt{LMHyper} reads both a $\forall\exists^*$ HyperLTL formula $\varphi$ and LTL formula $\psi$ and searches for an (un)satisfiability proof for $(\psi, \varphi)$.
Internally, we represent the current candidate as a generalized Büchi automaton and use \texttt{spot} \cite{DuretLutzLFMRX16} to perform automata operations.
The only other available tool for $\forall\exists^*$ HyperLTL satisfiability is \texttt{MGHyper}~\cite{FinkbeinerHH18}, which implements the incremental approach to find models of bounded size.

\subsection{Random Benchmarks}

We compare \texttt{LMHyper} against \texttt{MGHyper} on random formulas where we sample the LTL body of a formula using \texttt{randltl} \cite{DuretLutzLFMRX16}.
The results are given in \Cref{fig:randomEval}.
On our benchmarks, \texttt{LMHyper} usually takes longer than \texttt{MGHyper} but can handle a larger percentage of formulas.
We observe that randomly generated HyperLTL formulas are, in most cases, satisfiable by a model with a single trace, as the atomic propositions are seldom shared between different trace variables. 
This explains the high success rate of \texttt{MGHyper} (see \cite{FinkbeinerHH18}) even though  \texttt{MGHyper} cannot prove unsatisfiability.

\begin{table}
   \caption{Comparison of \texttt{LMHyper} and \texttt{MGHyper} on 100 random formulas generated with \texttt{randltl} \cite{DuretLutzLFMRX16}. Size refers to the size of the AST, $p$ is the percentage of solved formulas,
   	$t$ the average time spent on solved cases in milliseconds, and \#Iter is the average number of iterations (number of recursive calls) used by \texttt{LMHyper}.
   	The timeout is set to 5sec.
   }
   \begin{tabular}{c@{\hspace{8mm}}cc@{\hspace{8mm}}ccc}
       \toprule
        & \multicolumn{2}{@{}c@{\hspace{8mm}}}{\texttt{MGHyper}} & \multicolumn{3}{c}{\texttt{LMHyper}} \\
       Size & $p$ & $t$ & $p$ & $t$ & \#Iter \\
       \midrule
       15& $95\%$ & $40$ & $100\%$  &$235$ & $0.38$   \\
       16& $93\%$ & $39$ & $99\%$  &$239$ & $0.44$  \\
       17& $95\%$ & $39$ & $100\%$  &$221$ & $0.43$  \\
       18& $92\%$ & $38$ & $100\%$ &$201$ & $0.39$   \\
       19& $95\%$ & $40$ & $100\%$  &$180$ & $0.43$  \\
       20& $97\%$ & $42$ & $100\%$  &$215$ & $0.27$  \\
       \bottomrule
   \end{tabular}
    \label{fig:randomEval}
\end{table}

\begin{table}
    \caption{Comparison of \texttt{LMHyper} and \texttt{MGHyper} on hand-crafted specifications. We give the result (\cmark{} if the specification is satisfiable and \xmark{} if it is unsatisfiable), the time in ms, and the number of iterations needed by \texttt{LMHyper}. The timeout is set to 5min.}\label{fig:selectEval}
    \begin{tabular}{c@{\hspace{8mm}}cc@{\hspace{8mm}}ccc}
        \toprule
         & \multicolumn{2}{@{}c@{\hspace{8mm}}}{\texttt{MGHyper}} & \multicolumn{3}{c}{\texttt{LMHyper}} \\
        Problem & Res & $t$ & Res & $t$ & \#Iter \\
        \midrule
        Inf & - & TO & \cmark{} &350 & 1  \\
        \Cref{ex:introEx} & - & TO & \cmark{} &232 & 1 \\
        Enforce-2& \cmark{} & 444 & \cmark{} &262 & 0  \\
        Enforce-3& - & TO & \cmark{} &334 & 0  \\
        Enforce-5& - & TO & \cmark{} &491 & 0  \\
		Unsat-3& - & TO & \xmark{} &777 & 3  \\
        Unsat-5& - & TO & \xmark{} &1363 & 5  \\
        Unsat-9& - & TO & \xmark{} &1681 & 9  \\
        \bottomrule
    \end{tabular}
\end{table}

\subsection{Infinite and Large Models}
We compiled a small number of more interesting properties that do not have single-trace models. Our results are depicted in \Cref{fig:selectEval}.
The Inf specification expresses that a model has infinitely many traces.
\Cref{ex:introEx} is the example from the introduction.
The Enforce-$n$ specification enforces a model that has at least $n$ traces.
It is defined as $\forall \pi \exists \pi_1\ldots \pi_n\ldot \bigwedge_{i \neq j} \ltlF (a_{\pi_i} \not \leftrightarrow a_{\pi_j})$.
The Unsat-$n$ specifications are unsatisfiable. Their definition is a trace property $\psi \coloneqq (\neg a) \ltlU (a \land \ltlN\ltlG \neg a) \land \ltlN^n \ltlG\neg a$ combined with the hyperproperty $\varphi \coloneqq \forall \pi \exists \pi' \ldot \ltlF (a_\pi \land \ltlN a_{\pi'})$. The formula is designed such that \Cref{alg} requires $n$ iterations to discover unsatisfiability. 
\texttt{MGHyper} times out for most of the examples; even on simple properties like Enforce-$3$.
In contrast, \texttt{LMHyper} can verify properties enforcing many traces in a single iteration because the number of iterations is independent of the number of traces in a model.
As expected, Unsat-$n$ is unsatisfiable and \texttt{LMHyper} requires multiple iterations to show this.

\section{Conclusion}

We have studied the satisfiability problem for $\forall^*\exists^*$ HyperLTL formulas in combination with LTL formulas describing functional behavior.
To obtain results below the general $\Sigma_1^1$ complexity of HyperLTL, we have focused on simpler hyperproperties belonging to the classes of temporal safety and temporal liveness as well as fragments thereof.
We have shown that temporal safety is an expressive class that is very well suited for satisfiability studies and enjoys \coRE-completeness.
In combination with general LTL properties, already very simple formulas like invariants cause $\Sigma_1^1$-completeness.
The temporal liveness class, on the other hand, is $\Sigma_1^1$-complete in general but contains non-trivial fragments that are decidable, even in combination with arbitrary LTL formulas.

We have shown that functional specifications given in LTL play a significant role in the undecidability of general hyperproperties.
The main open question for future work is whether further decidable fragments can be found by restricting the operator structure of the functional specification.

\begin{acks}
	All authors are partially supported by the \grantsponsor{DFG}{\emph{German Research Foundation} (DFG)}{} in project \grantnum{DFG}{389792660, TRR 248} (Center for Perspicuous Systems).
	M.~Krötzsch is additionally supported by the \grantsponsor{BMBF}{\emph{Bundesministerium für Bildung und Forschung} (BMBF)}{} in project \grantnum{BMBF}{ScaDS.AI } (Center for Scalable Data Analytics and Artificial Intelligence), and by the Center for Advancing Electronics Dresden (cfaed).
	R.~Beutner, B.~Finkbeiner and J.~Hofmann are additionally supported by the \grantsponsor{ERC}{\emph{European Research Council} (ERC)}{} in project \grantnum{ERC}{OSARES (No.~68330)}.
	R.~Beutner and J.~Hofmann carried out this work as members of the Saarbrücken Graduate School of Computer Science.
\end{acks}

\bibliographystyle{ACM-Reference-Format}

{
	\interlinepenalty=10000 % to prevent page breaks in citations leading to problems with hyperref
\bibliography{references}
}

\iffullversion

\appendix

\section{Additional Material For \Cref{sec:safety}}\label{app:secSafety}

\begin{definition}
    A \emph{nondeterministic Turing machine} is a tuple $T = (Q, Q_0, \Gamma, \delta, F)$ where $Q$ is a finite set of states, $Q_0 \subseteq Q$ the initial states, $\Gamma$ a finite alphabet, $\delta \subseteq (Q \times \Gamma) \times (Q \times \Gamma \times \{L, R\})$ the transition relation and $F \subseteq Q$ a set of accepting states. 
    A transition $((q, a), (q', a', d) )\in \delta$ means that the if the TM is in state $q$ and reads $a$, it updates its state to $q'$, writes $a'$, and moves either to the left ($d = L$) or the right ($d = R$). 
    We assume a dedicated blank symbol $\# \in \Gamma$.
    In an initial configuration for a finite word $w \in \Gamma$ (not containing $\#$), the head is at position $0$, the state is in some state in $Q_0$ and the tape contains $w$ (followed by infinitely many $\#$).
    We say that a TM accepts the empty word if there is a run starting in an initial configuration that eventually visits a configuration where the state is in $F$.
    A TM is \emph{deterministic} if $Q_0 = \set{q_0}$ and for each $(q,a)$, there exists at most one successor in $\delta$.
\end{definition}

\coREContainment*
\begin{proof}
	We already gave the construction of a FOL formula $\Theta$ in the proof sketch in the main part (in \Cref{sec:safety}).
	Here we only show its correctness, i.e., we show that the FOL formula $\Theta$ and HyperLTL formula $\varphi$ are equisatisfiable.
	
	Assume $\Theta$ is satisfiable and fix a first-order model. Let the set $X$ be the set of elements from the $\mathit{Trace}$ domain which may be assigned to some variable $x_i$ in any possible evaluation of the quantifiers. We iteratively construct a trace for any element of $X$.
	To do so, let $i_0, i_1, \ldots$ be a fixed sequence of element from $\mathit{TimePoint}$ such that $\mathit{Succ}(i_j, i_{j+1})$ for any $j \in \nat$ and $i_0$ is the constant described above. This sequence might not be unique and elements might occur several times but we need to fix one such sequence to obtain well-defined traces.
	For each element $v \in X$, we define a trace $t_v \in \Sigma^\omega$, by setting $t_v(n) \coloneqq \set{a \mid P_a(v, i_n)}$ for every $n \in \nat$, where we write $P_a(v, i_n)$ whenever this holds in the fixed firs-order model of $\Theta$.
	It is easy to see that $T \coloneqq \{t_v \mid v \in X\}$ is a model of $\varphi$.
	This holds as $\Theta$ minims the quantification in $\varphi$, so whenever a quantifier is instantiated with $v \in X$, we use trace $t_v \in T$ for the respective quantifier in $\varphi$.  
	By construction of $\rho_{q}$ we ensure that $\mathcal{A}_\phi$ has an accepting run on all tuples of $n$ traces chosen in the quantifier prefix. 
	
	For the other direction, assume that $\varphi$ is satisfiable by trace set $\traceSet$, which we choose as domain for sort $\mathit{Trace}$. For $\mathit{TimePoints}$, we choose the set of natural numbers with $i_0 = 0$. We set $P_a(t, i)$ to true iff $a \in t[i]$. For every assignment of $t_1, \ldots t_n$ to the trace variables, we fix a run through $\mathcal{A}$ and set $\mathit{State}_q(t_1, \ldots, t_n, i)$ to true iff the the run is in state $q$ in step $i$. The resulting structure satisfies $\phi$.
\end{proof}

\coReHardness*
\begin{proof}   
    We reduce from the non-halting problem of deterministic Turing machines on the empty word, which is \coRE complete. We assume, w.l.o.g., that the tape of the Turing machine is left-bounded and only takes a step to the left when it is possible. 
    We encode the position of the head with a proposition $h$. Throughout the construction, we maintain the invariant on all relevant traces that $h$ is set exactly once on the trace. We cannot encode this property directly as it would require to nest multiple $\LTLglobally$ operators.
    We encode the alphabet $\Gamma$ and the set of states $Q$ with sets of atomic propositions, one for each symbol and enforce that in each step, exactly one for each set holds, for $S = \Gamma$ or $S = Q$:
    \begin{align*}
        \forall \pi \ldot \LTLglobally \, \Big( \bigvee_{a \in S} \big(a_\pi \land \bigwedge_{a \neq b \in S} \neg b_\pi\big) \Big)
    \end{align*}
    We fix the current state to be the one that holds in the position of the head.
    Initially, the TM is in state $q^0$, the head at position 0, and the tape is blank. 
    We require that the initial configuration is present in the set:
    \begin{align*}
        \exists \pi \ldot h_\pi \land q^0_\pi \land \LTLglobally (\#_\pi \land \LTLnext \neg h_\pi)
    \end{align*}
    Now we encode the possible transitions with a $\forall \pi \exists \pi'$ formula. We ensure that if the configuration encoded by $\pi$ is a valid one (i.e., $h$ only holds once), then the successor configuration is also valid. 
    For correct transitions, all positions on $\pi'$, which are not left or right of the head in $\pi$, must remain unchanged. Second, the head must move either left or right and the symbol and state propositions are only allowed to change in the position of the old head.
    \begin{alignat*}{2}
        &\psi_\text{LxorR} \coloneqq && \, \bigg((\neg h_\pi \land \LTLnext (\neg h_\pi \land \LTLnext \neg h_\pi)) \rightarrow \LTLnext \bigwedge_{a \in \ap} a_\pi \leftrightarrow a_{\pi'}\bigg) \\
        & && \land \LTLnext h_\pi \rightarrow \bigg( \LTLnext \neg h_{\pi'} \land (h_{\pi'} \oplus \LTLnext \LTLnext h_{\pi'}) \\
        & && \qquad \land \big(\bigwedge_{a \in \Gamma} (a_\pi \leftrightarrow a_{\pi'}) \land  \LTLnext \LTLnext (a_\pi \leftrightarrow a_{\pi'})\big)\bigg) 
    \end{alignat*}
    Now we translate each transition $(q, a), (q', a', L) \in \delta$ as follows. If the head moves right instead of left, we change the position of the $\LTLnext$ operator accordingly.
    \begin{align*}
        &\psi_{\text{trans}_1} \coloneqq \LTLnext (h_\pi \land a_\pi \land q_\pi) \rightarrow h_{\pi'} \land q'_{\pi'} \land \LTLnext a'_{\pi'}
    \end{align*}
    The final transition formula is the following:
    \begin{align*}
        &\forall \pi \ldot \exists \pi' \ldot \LTLglobally (\psi_\text{LxorR} \land ( \psi_{\text{trans}_1} \lor \psi_{\text{trans}_2} \lor \cdots))
    \end{align*}
    We only encode transitions where $q' \notin F$, i.e. only those transitions that do not make the TM halt.
    The conjunction of the above formulas can be easily transformed into a $\forall \exists^2$ formula with a single $\LTLglobally$ and only non-nested $\LTLnext$ operators in the scope of $\LTLglobally$.
    Now, the TM has an infinite non-halting run iff the conjunction of the above formulas has a satisfying model. The trace set might not only contain the witnessing run but also non-valid configurations or non-reachable configurations, which we can just ignore. 
\end{proof}

\safetySigma*
\begin{proof}
    Membership in $\Sigma^1_1$ follows from~\cite{FortinKT021}.
    We reduce from the recurring computation problem of nondeterministic two-counter machines. 
    The key idea is to use LTL formulas over pairs of encoded configurations and use the hyperproperty only to ensure that every second component matches with some first component of another trace.
    We thus set $\ap = \ap^1 \cup \ap^2$ and each $\ap^i = \set{\mathbf{l}_1^i, \ldots, \mathbf{l}_n^i, \mathbf{c}_1^i, \mathbf{c}_2^i, \mathbf{t}^i}$.
    The counter values for counter $c_1, c_2$ are encoded by two atomic propositions $\mathbf{c}_1$ and $\mathbf{c}_2$ which hold exactly once, i.e., counter $x$ has value $v$ if $\mathbf{c}_x$ occurs in the $v$th position. We use $n$ atomic propositions for the instructions, and require that in the first position, exactly one instruction label holds.
    Note that all of the above requirements for a valid configuration are easily expressible in LTL.
    To ensure that the computation visits the initial instruction $l_1$ infinitely often, we use a proposition $\mathbf{t}$, which must also hold exactly once and decreases with every computation step. When $t$ is in position $0$, $\mathbf{l}_1$ must hold.
    We encode this requirement as $\mathbf{l}_1^1 \lor \LTLglobally (\mathbf{t}^2 \leftrightarrow \LTLnext \mathbf{t}^1)$.
    We also encode that each trace denotes an update step of the 2CM.
    For example instruction $l_i = \big[c_1 \coloneqq c_1+1 \texttt{ goto } \{l_j, l_{k}\}\big]$ can be encoded as:
    \begin{align*}
        \mathbf{l}_i^1 \rightarrow \big[ \ltlG ((\mathbf{c}_2^1 \leftrightarrow \mathbf{c}_2^2) \land (\mathbf{c}_1^1 \leftrightarrow \LTLnext \mathbf{c}_1^2)) \land (\mathbf{l}_j^2 \lor \mathbf{l}_k^2) \big]
    \end{align*}
    Decrement instructions are analogous.
    Instructions of the form $l_i = \big[\texttt{if } c_1 = 0 \texttt{ then goto } l_j \texttt{ else goto } l_k\big]$ can be encoded as.
    \begin{align*}
        \mathbf{l}_i^1 \rightarrow & \Big[(\mathbf{c}_1^1 \rightarrow \mathbf{l}_j^2) \land (\neg \mathbf{c}_1^1 \rightarrow \mathbf{l}_k^2) \land \bigwedge_{x \in {1,2}} \LTLglobally ( \mathbf{c}_x^1 \leftrightarrow \mathbf{c}_x^2 )\Big]
    \end{align*}
    We then take the conjunction over the resulting formulas for each instruction.
    
    What is left to state is the hyperproperty  $\varphi$, which matches second components to first components.
    \begin{align*}
        \forall \pi \ldot \exists \pi' \ldot \LTLglobally \bigwedge_{a \in \ap} a^2_{\pi} \leftrightarrow a^1_{\pi'}
    \end{align*}
    Additionally, the initial configuration must be present in the first component of a trace.
    \begin{align*}
        \exists \pi \ldot (\mathbf{c}_1^1)_\pi \land (\mathbf{c}_2^1)_\pi \land (\mathbf{l}_1^1)_\pi
    \end{align*}
    The resulting specification $(\psi, \varphi)$ is satisfiable iff the 2CM has a recurring computation starting in the initial state. The formula $\varphi$ is from the $\forall\exists^2$ fragment and uses only a single $\ltlG$ without nested temporal operators.
\end{proof}

\propTheo*
\begin{proof}
    Let $\varphi = Q^* \ldot \phi$.
    For membership in \NEXPT, let $k$ be the lookahead needed to evaluate $\phi$, which is upper-bounded the number of $\ltlN$ operators occurring in $\phi$. 
    Now define $M \subseteq \Sigma^k$ as all finite traces $\tau$ of length $k$ such that $\tau$ can be extended to a trace satisfying $\psi$. 
    $M$ can be constructed in exponential time by converting $\psi$ to a Büchi automaton and checking if each $\tau \in M$ has an accepting run in linear time. 
    We then nondeterministically guess a subset $M' \subseteq M$ and verify that $M' \models \varphi$ which can easily be done in nondeterministic exponential time, giving the desired \NEXPT-upper bound. 
    
    To show hardness we reduce the acceptance of a time-bounded nondeterministic Turing machine on the empty word to a HyperLTL formula (our construction does not require any temporal operators, not even $\ltlN$s).
    Consequently, only the first position of each trace is relevant.
    Let $T = (Q, Q_0, \Gamma, \delta, F)$ be a exponential time-bounded nondeterministic TM.
    As our formula uses no $\ltlN$s, only the first position is relevant so we can see each trace as a propositional evaluation over $\ap$.
    The idea is that each trace in the model encodes a piece of information of an accepting run of $T$.
    Such a piece of information is a tuple $(s, p, \gamma, q)$ where $s, p \in \nat$, $\gamma \in \Gamma$ and $q \in Q \cup \{\bot\}$.
    The information encodes that in time-step $s$ and at position $p$, the tape content is $\gamma$ and the head is at that position and the machine in state $q$ if $q \in Q$ or the head is not at that position (if $q = \bot$).
    As $T$ is time (and thus space) bounded, $s$ and $p$ are bound by $2^n$ for some $n$.
    
    In our formula we now use propositions $\vec{\mathbf{s}} = \mathbf{s}^1, \ldots, \mathbf{s}^{n}$ to encode a time step and $\vec{\mathbf{p}} = \mathbf{p}^1, \ldots, \mathbf{p}^{n}$ to encode a position on the tape (both as a binary counter). 
    Moreover, for each state $q \in Q \cup \{\bot\}$ and each letter $\gamma \in \Gamma$, we introduce an atomic proposition $q$ and $\mathbf{\gamma}$. 
    In each evaluation exactly one of the tape-content propositions holds and at most one state-proposition holds (if none holds, the head is at a different step).

    Using HyperLTL, we specify that the information encoded in the traces is consistent and forms an accepting run of the TM:\\
    \textbf{(1)} At most one of the state-propositions and exactly one of the tape-propositions holds:
    \begin{align*}
        &\forall \pi\ldot\left( \bigvee_{\gamma \in \Gamma} \gamma_\pi \land \bigwedge_{\gamma \neq \gamma' \in \Gamma} \neg \gamma'_\pi\right) \land \bigwedge_{q \neq q' \in Q} \neg (q_\pi \land q'_\pi)
    \end{align*}
    \textbf{(2)} If two traces agree in time step and position they also agree on tape content, state (and thus head position):
    \begin{align*}
        &\forall \pi \forall \pi' \ldot \left( \bigwedge_{i = 1}^n (\mathbf{s}^i_{\pi} \leftrightarrow \mathbf{s}^i_{\pi'}) \land (\mathbf{p}^i_{\pi} \leftrightarrow \mathbf{p}^i_{\pi'})\right) \rightarrow \\
        &\quad\quad\quad\quad\quad\left(\bigwedge_{q \in Q} (q_\pi \leftrightarrow q_{\pi'}) \land \bigwedge_{\gamma \in \Gamma} (\gamma_\pi \leftrightarrow \gamma_{\pi'})\right)
    \end{align*}
    \textbf{(3)} The TM state (and therefore also head position) agrees across all traces that are in the same time step.
    Phrased differently, for any two different traces that agree on the step and the automaton state, are already equal. 
    Here we write $\pi = \pi'$ to mean that $\pi$ and $\pi'$ agree (in the first position)
    \begin{align*}
        \forall \pi \forall \pi' \ldot \left(\bigwedge_{i = 1}^n (\mathbf{s}^i_{\pi} \leftrightarrow \mathbf{s}^i_{\pi'}) \land  \bigwedge_{q \in Q} (q_\pi \leftrightarrow q_{\pi'}) \right) \rightarrow  \pi = \pi'
    \end{align*}
    \textbf{(4)} The initial configuration of the TM is in the model. 
    We cannot encode this with existential quantification directly, as we would need to refer to exponentially many positions and therfiore also exponentially many traces.
    Instead, we express it as a $\forall\exists$ property. 
    We use syntactic sugar and write $\vec{\mathbf{p}}_\pi = \vec{\mathbf{p}}_{\pi'} + 1$ to indicate that the binary counter increases by one position. This can easily be encoded in propositional logic in quadratic size. 
    For a constant $c$ we write $\vec{\mathbf{p}}_\pi = c$ to denote the binary counter given by $\vec{\mathbf{p}}$ equals $c$.
    We express the desired property as a conjunction of the following two formulas.
    The first formula is given by:
    \begin{align*}
        & \exists \pi \ldot \vec{\mathbf{s}}_\pi = 0 \land  \vec{\mathbf{p}}_\pi = 0 \land \#_{\pi} \land \bigvee_{q \in Q_0} q_{\pi}
    \end{align*}
    Which expresses that at step $0$ at least the very first position is filled with $\#$ and the head is at that poistion (in an initial state).
    The second information propagates this information to all remaining positions in the first step:
    \begin{align*}
        \forall \pi \exists \pi'\ldot &\big(\vec{\mathbf{s}}_\pi = 0 \land \vec{\mathbf{p}}_\pi \neq 2^n-1 \big) \rightarrow \\
        &\quad\quad\big( \vec{\mathbf{s}}_{\pi'} = 0 \land \vec{\mathbf{p}}_{\pi'} = \vec{\mathbf{p}}_{\pi} + 1 \land \#_{\pi'}\big)
    \end{align*}
	Note that \textbf{(3)} already implies that the position is set consistently in the same step, so we do not need to require that $\bot_{\pi'}$ holds.  \\
    \textbf{(5)} The model encodes valid steps of the TM.
    As the transitions of a TM are local, it is sufficient to compare three consecutive tape positions with the same three tape positions in the next step. 
    Thus, for every three traces that share the same time step and encode consecutive positions, there exists three traces on the next time step with the same positions such that the step is possible in the TM. 
    The resulting $\forall^3\exists^3$ formula then consists of a finite conjunction over those valid ``transitions triples'' of $T$.
    We do not give the formula explicitly, as the notation gets clustered. \\
    \textbf{(6)} The TM is eventually in an accepting state which we can express via $\exists \pi. \bigvee_{q \in F} q_\pi$.	

    Note that because of the explicit step count (via $\vec{\mathbf{s}}$), each model directly encodes a unique run of the TM.
    In particular, there can not be spurious facts that do not hold on the defined run of the TM. 
    Consequently, we get that the resulting formula is satisfiable iff TM has a $2^n$-time bounded accepting run. 
    The reduction outlined above can be done in logspace, giving the desired \NEXPT-hardness.
\end{proof}

\section{Additional Material for \Cref{sec:livness}}
\label{app:liveness}

\npEventual*
\begin{proof}
	We first consider the case where we have only a single eventuality.
    Let $\varphi = \forall \pi_1 \ldots \pi_n \exists \pi_{n+1} \ldots \pi_{n+m} \ldot \ltlF \phi$.
    
    We claim that we can actually assume that $\phi$ contains no $\ltlN$s operators. 
    We assume that in $\phi$ the $\ltlN$s occur directly in from of atomic proportions, i.e., $\phi$ is a boolean formula over atoms of the form $\ltlN^n a_{\pi_i}$.
    Transforming a formula in this form is possible in polynomial time. 
    Now let $k$ be the lookahead needed to evaluate $\phi$ (which is upper bounded by the number of $\ltlN$s in $\phi$).
    We define $\phi'$ as the formula over $\ap_0 \cup \cdots \cup \ap_k$ (where $\ap_i = \{a_i \mid a \in \ap\}$) obtained from $\phi$ by replacing each atom $\ltlN^n a_{\pi_i}$ with the atomic proposition $(a_n)_{\pi_i}$.
    As $\phi$ occurs under a $\ltlF$, it is easy to see that $\forall \pi_1 \ldots \pi_n \exists \pi_{n+1} \ldots \pi_{n+m} \ldot \ltlF \phi$ is satisfiable iff $\forall \pi_1 \ldots \pi_n \exists \pi_{n+1} \ldots \pi_{n+m} \ldot \ltlF \phi'$ is satisfiable.
    
    So let us assume in the following that $\phi$ contains no temporal operators, i.e., $\phi$ is a propositional boolean formula over $\ap_{\pi_1} \cup \cdots \cup \ap_{\pi_{n+m}}$.
    We claim that $\varphi$ is satisfiable iff
    \begin{align*}
    	\zeta \coloneqq\phi \land \bigwedge_{i = 2}^n \bigwedge_{a \in \ap} (a_{\pi_i} \leftrightarrow a_{\pi_1})
    \end{align*}
    is satisfiable (viewed as a propositional formula).
    For the left to right direction, assume that $\varphi$ is satisfiable by $\traceSet$. 
    We now pick any $t \in \traceSet$ and instantiate all universal quantifiers with $t$. 
    Then let $t_{n+1}, \ldots, t_{n+m} \in \traceSet$, be witness traces for the existential quantification.
    So $[\pi_1 \mapsto t, \ldots, \pi_n \mapsto t, \pi_{n+1} \mapsto t_{n+1}, \ldots, \pi_{n+m} \mapsto t_{n+m}] \models \ltlF \phi$.
    Let $n \in \nat$ be such that $[\pi_1 \mapsto t, \ldots, \pi_n \mapsto t, \pi_{n+1} \mapsto t_{n+1}, \ldots, \pi_{n+m} \mapsto t_{n+m}][n,\infty] \models \phi$.
    We construct an assignment $\alpha : \ap_{\pi_1} \cup \cdots \cup \ap_{\pi_{n+m}} \to \mathbb{B}$ as follows:
    For each $a_{\pi_i}$ we define $\alpha(\pi_i) \coloneqq t(n)(a)$ if $1 \leq i \leq n$ and $\alpha(\pi_i) \coloneqq t_i(n)(a)$ if $n+1 \leq i \leq n+m$.
    It is easy to see that $\alpha$ satisfies $\zeta$ (note that all universally quantified copies are chosen by $t$ so $\bigwedge_{i = 2}^n \bigwedge_{a \in \ap} (a_{\pi_i} \leftrightarrow a_{\pi_1})$ holds trivially).
    
    For the reverse, assume that $\zeta$ is satisfiable and let $\alpha$ be a satisfying variable assignment over $\ap_{\pi_1} \cup \cdots \cup \ap_{\pi_{n+m}}$. 
    We construct a model with $n+m$ traces $t_1, \ldots, t_{n+m}$ as follows:
    For every indices $i_1, \ldots, i_n \in \{1, \ldots, n+m\}$ we choose $m$ distinct indices $j_1, \ldots, j_m \in  \{1, \ldots, n+m\}$ that are also different from all $i_1, \ldots, i_n$ (this is possible as we fixed $n+m$ traces).
    Now fix some \emph{fresh} position $n \in \nat$ and ensure that $[\pi_1 \mapsto t_{i_1}, \ldots, \pi_n \mapsto t_{i_n}, \pi_{n+1} \mapsto t_{j_1} \ldots \pi_{n+m} \mapsto t_{j_m}][n, \infty] \models \phi$.
    This is always possible as we can use $\alpha$ to construct the assignments of all involved traces at position $k$.
    Here it is crucial that we added the second conjunct $\bigwedge_{i = 2}^n \bigwedge_{a \in \ap} (a_{\pi_i} \leftrightarrow a_{\pi_1} )$ to $\zeta$ as this ensures that we construct a position at which the property holds even if  $t_{i_1}, \ldots, t_{i_n}$  are not distinct. 
    We iterate this for every of the finitely many elements in $\{1, \ldots, n+m\}^n$ (always choosing a fresh position $n$), thereby defining a model of $\varphi$. 

    To show \texttt{NP}-hardness, we can easily reduce from propositional SAT problem:
    A formula $\zeta$ over $\ap$ is satisfiable iff $\exists \pi. \ltlF \zeta_\pi$ is satisfiable (we write $\zeta_\pi$ for $\zeta$ when replacing all atoms $a$ with $a_\pi$).

    To handle multiple eventualities, reuse the above proof. 
    Let $\varphi = \forall^* \exists^* \ldot \ltlF \phi_1 \land \cdots \land \ltlF \phi_n$. 
    We again assume, w.l.o.g., that each $\phi_i$ contains no $\ltlN$s.
    We introduce a different copy of $\ap_i$ for the $i$th $\ltlF$ operator. In each $\ltlF$-subformula, the propositional variables are exchanged accordingly, i.e., let $\phi_i'$ be $\phi_i$ where the atomic propositions from $\ap$ are replaced with those from $\ap_i$ .
   	It is easy to see that $\varphi$ is satisfiable iff $\forall^* \exists^* \ldot \ltlF (\phi'_1 \land \cdots \land \phi'_n)$ is satisfiable, as we can simply fulfill all eventualities at different timepoints. 
\end{proof}

\begin{lemma}\label{lem:npEventualWithProp}
	The satisfiability problem is \NEXPT-complete for specifications $(\top, \varphi)$ where $\varphi$ is of the form $\forall^* \exists^* \ldot \ltlF (\ltlN^*) \land \cdots \land \ltlF (\ltlN^*)$ (and we do allow propositional formulas occur outside of the $\ltlF$ operators).
\end{lemma}
\begin{proof}
	We prove this lemma by combining \Cref{lem:npEventual} and \Cref{theo:prop}.
	The crux is that we can separate the propositional formulas outside the $\ltlF$.
	Assume we are given a formula $\varphi = \forall^* \exists^*\ldot (\ltlF \phi_1) \land \cdots \land (\ltlF \phi_n) \land \phi'$ where $\phi_1, \ldots, \phi_n, \phi'$ contain only $\ltlN$s.
	It it easy to see that $\varphi$ is satisfiable iff both $\forall^*\exists^*\ldot (\ltlF \phi_1) \land \cdots \land (\ltlF \phi_n)$ and $\forall^* \exists^*\ldot \phi'$ are satisfiable.
	%That is, a formula $\forall^*. \exists^*. (\ltlF \phi_1) \land \cdots \land (\ltlF \phi_n) \land \phi'$ where $\phi_1, \ldots, \phi_n, \phi'$ contain only $\ltlN$s is satisfiable iff both $\forall^*. \exists^*. (\ltlF \phi_1) \land \cdots \land (\ltlF \phi_n)$ and $\forall^*. \exists^*. \phi'$ are satisfiable.
	The former is decidable in \texttt{NP} (see \Cref{lem:npEventual}) and the latter in \NEXPT (see \Cref{theo:prop}), so the \NEXPT upper bound follows.
	For the lower bound, we can directly use the hardness shown in \Cref{theo:prop} and ignore the ability to use $\ltlF$ formulas. 
\end{proof}

\lemmaElimNext*
\begin{proof}
	Let $\varphi = \forall^n\exists^m\ldot \ltlF \phi$.
	Let $k$ be the lookahead needed to evaluate $\phi$ (which is upper bounded by the number of $\ltlN$s in $\phi$).
	Similar to the proof of \Cref{lem:npEventual} we eliminate $\ltlN$s by using letting traces range over tuples.
	
	We construct a formula $\phi'$ over $\mathit{AP}_0 \cup \cdots \cup \mathit{AP}_k$ (where $\ap_i = \{a_i \mid a \in \ap\}$)  as follows:
	We assume that in $\phi$ the $\ltlN$s occur directly in from of atomic proportions, i.e., $\phi$ is a boolean formula over atoms of the form $\ltlN^n a_{\pi_i}$.
	Transforming a formula in this form is possible in polynomial time.  
	We then define $\phi'$ as the formula over $\ap_0 \cup \cdots \cup \ap_k$ obtained from $\phi$ by replacing each atom $\ltlN^n a_{\pi_i}$ with the atomic proposition $(a_n)_{\pi_i}$.
    
    The traces we consider now range over letters from $\mathit{AP}_0 \cup \cdots \cup \mathit{AP}_k$, i.e., over the window of the next $k+1$ steps.
    In the LTL part of the specification we assert that this window is consistent, i.e., the $i+1$th state (evaluation over $\mathit{AP}_{i+1}$ equals the evaluation of the $i$th state in the next step).
    We define $\psi'$ as 
    \begin{align*}
        \psi' \coloneqq  \psi_0 \land \ltlG \bigwedge_{i = 0}^{k-1}  \bigwedge_{a \in \mathit{AP}} \big(a_{i+1} \leftrightarrow \ltlN a_i\big)
    \end{align*}%
	where $\psi_0$ is obtained from $\psi$ by replacing each proposition $a \in \ap$ with $a_0 \in \ap_0$.
	We define $\varphi' = \forall^n\exists^m\ldot \ltlF \phi'$.
    It is easy to see that $(\psi, \varphi)$ and $(\psi', \varphi')$ are equisatisfiable. 
\end{proof}

\alwaysSat*
\begin{proof}
	Let $\varphi = \forall \pi_1\exists \pi_2\ldots \pi_{m+1}\ldot \phi$ and let $\aut_\phi$ be a deterministic Büchi automaton for $\phi$ over $\ap_{\pi_1} \cup \cdots \cup \ap_{\pi_{m+1}}$.
	For finite traces $u_1, \ldots, u_{m+1}$ of the same length (say $k$), we define $\mathit{zip}(u_1, \ldots, u_{m+1})$ as the finite trace (of length $k$) over $\ap_{\pi_1} \cup \cdots \cup \ap_{\pi_{m+1}}$ that combines $u_1, \ldots, u_{m+1}$, i.e., the evaluation of $\ap_{\pi_i}$ is copied from $t_i$.
    The crucial property we use is that for a deterministic liveness property, we can always revisit an accepting state even after having read an arbitrary finite word.
    We claim that there always exists a finite model of size at most $m+1$.
    For any $i \in \{1, \ldots, m+1\}$ define $f(i)$ as the vector $(1, \ldots, i-1, i+1, \ldots, m+1)$ (which has length $m$).
    We iteratively constructs a model as follows:
    Initially, we set $u_1, \ldots, u_{m+1} = \epsilon$. 
    For each $j = 0, 1, 2, \ldots$ let $i = (j \% (m+1)) + 1$ (this way we consider each $i \in \{1, \ldots, m+1\}$ infinitely many times) and $(j_1, \ldots, j_m)  = f(i)$.
    We now extend each of the traces in $\traceSet$ by some finite, non-empty word $u'_1, \ldots, u'_{m+1}$ of the same length such that $\mathit{zip}(u_{i} u'_{i}, u_{j_1} u'_{j_1}, \ldots, u_{j_m} u'_{j_m})$ reaches an accepting state in $\aut_\phi$.
    As $\aut_\phi$ is a deterministic liveness automaton this is always possible. 
    
    Let $t_1, \ldots t_{m+1}$ be the infinite traces constructed in the limit.
    It is easy to see that $\traceSet = \{t_1, \ldots, t_{m+1}\}$ is a model of $\varphi$.
    In the limit, this constructs traces of infinite length which serve as a model for $\varphi$.
    For each trace $t_i$, the traces with index determined by $f(i)$ can be chosen as witness traces for existential quantification.
    By construction, the (unique) run of $\aut_\phi$ on the resulting tuple of traces is accepting.
\end{proof}

\section{Additional Material for \Cref{sec:largestModelsAlg}}

\algInv*
\begin{proof}
    Initially, the property obviously holds as $\mathcal{A}$ encodes $\varphi$ and every trace which occurs in any satisfying trace set must have a run through $\mathcal{A}$ as the $\forall$ trace.
    It remains to show that $\mathcal{A}_\text{new}$ does not exclude traces that occur in satisfying $T$.
    The construction ensures that if $t \in \lang{\mathcal{A}^\forall}$ but $t \not \in \lang{\mathcal{A}^\forall_\text{new}}$, then there exists no witness trace $t'$ such that $t_\pi \cup t'_{\pi'} \in \lang{\mathcal{A}}$ and $t' \in \lang{\mathcal{A^\forall}}$.
    By induction, $\lang{\mathcal{A}}$ is a superset of the union of all satisfying trace sets, therefore, there can be no $T$ such that $t \in T$ and $T \models \forall \pi \exists \pi' \ldot \phi$.
\end{proof}

\fi

\end{document}
\endinput
